\documentclass[journal]{IEEEtran}
\ifCLASSINFOpdf
\else
\fi
\usepackage{cite}
\usepackage{amsmath,amssymb,amsfonts}
\usepackage{amsthm}
\usepackage{mathtools} 
\usepackage{algorithm,algorithmic}
\usepackage{graphicx}
\usepackage{textcomp}
\usepackage{subfigure}
\usepackage{xcolor}
\usepackage{bbm}

\usepackage{bm}
\usepackage{comment}
\usepackage{multirow}
\usepackage[printonlyused]{acronym}

\newcommand{\fref}[1]{Figure~\ref{#1}}
\newcommand{\sref}[1]{Section~\ref{#1}}
\newcommand{\tref}[1]{Table~\ref{#1}}
\usepackage[shortlabels]{enumitem}
\DeclareMathOperator{\proj}{proj}
\newcommand{\vct}{}
\newcommand{\vctproj}[2][]{\proj_{\vct{#1}}\vct{#2}}
\usepackage{tikz}
\usepackage{pgfplots}
\usepackage{array}
\newcolumntype{C}[1]{>{\centering\arraybackslash}p{#1}}
\usepackage{caption}
\usetikzlibrary{arrows,automata}
\pgfplotsset{compat=newest}
\usetikzlibrary{plotmarks}
\usetikzlibrary{arrows.meta}

\usepackage{grffile}
\usepgfplotslibrary{patchplots}
\tikzset{
    invisible/.style={opacity=0,text opacity=0},
    visible on/.style={alt={#1{}{invisible}}},
    alt/.code args={<#1>#2#3}{%
      \alt<#1>{\pgfkeysalso{#2}}{\pgfkeysalso{#3}} 
    },
}
\tikzset{onslide/.code args={<#1>#2}{%
  \only<#1>{\pgfkeysalso{#2}} 
}}
\setlist[itemize]{leftmargin=0.5cm}
\newtheorem{thm}{Theorem}
\newtheorem{lemma}{Lemma}
\newcommand\copyrighttext{%
  \footnotesize \textcopyright 2019 IEEE. Personal use of this material is permitted.  Permission from IEEE must be obtained for all other uses, in any current or future media, including reprinting/republishing this material for advertising or promotional purposes, creating new collective works, for resale or redistribution to servers or lists, or reuse of any copyrighted component of this work in other works.}
\newcommand\copyrightnotice{%
\begin{tikzpicture}[remember picture,overlay]
\node[anchor=south,yshift=5pt] at (current page.south) {\fbox{\parbox{\dimexpr\textwidth-\fboxsep-\fboxrule\relax}{\copyrighttext}}};
\end{tikzpicture}%
}
\begin{document}
%
	
	\title{Optimizing Adaptive Video Streaming in Mobile Networks via Online Learning}
%
%
%

\author{Theodoros~Karagkioules,
        Georgios~S.~Paschos,
        Nikolaos~Liakopoulos,
        Atillio~Fiandrotti,
        Dimitrios~Tsilimantos
        and~Marco~Cagnazzo
\thanks{T. Karagkioules (\emph{corresponding author}), N. Liakopoulos and D. Tsilimantos are with the Mathematical and Algorithmic Sciences Lab, Paris Research Center, Huawei Technologies France,  e-mail: firstname.lastname@huawei.com.}
\thanks{G. S. Paschos was with the Mathematical and Algorithmic Sciences Lab, Paris Research Center, Huawei Technologies France while research on this work was conducted, email: gpasxos@gmail.com.}
\thanks{T. Karagkioules, A. Fiandrotti and M. Cagnazzo are with T\'el\'ecom Paris, LTCI,  Institut Polytechnique de Paris, France, e-mail: firstname.lastname@telecom-paris.fr.}
}

%
%

\markboth{}%
{Karagkioules \MakeLowercase{\textit{et al.}}:Optimizing Adaptive Video Streaming in Mobile Networks via Online Learning}
%



\maketitle
\copyrightnotice
	\begin{abstract}
In this paper, we propose a novel algorithm for video rate adaptation in \ac{HAS}, based on online learning. The proposed algorithm, named \emph{Learn2Adapt (L2A)}, is shown to provide a \emph{robust} rate adaptation strategy which, unlike most of the state-of-the-art techniques, does not require parameter tuning, channel model assumptions or application-specific adjustments. These properties make it very suitable for mobile users, who typically experience fast variations in channel characteristics. Simulations show that \emph{L2A} improves on the overall \ac{QoE} and in particular the average streaming rate, a result obtained independently of the channel and application scenarios.
	\end{abstract}



%
\IEEEpeerreviewmaketitle

	\section{Introduction}\label{Intro}


\IEEEPARstart{v}{ideo} streaming accounts nowadays for more than 75\% of the global Internet traffic, a percentage projected to reach a striking 82\% by 2022 \cite{cisco2017}. To facilitate this increasing demand for video consumption, \ac{HAS} has been adopted as the main technology for video streaming over the Internet, gaining significant popularity as it allows video clients to seamlessly adapt to changing network conditions and video content to be distributed over existing web service infrastructures. In 2012, the \ac{MPEG} consortium created the \ac{DASH} standard \cite{Sodagar}, that has since become the dominant \ac{HAS} method.

According to \ac{DASH}, the video content is first encoded at multiple quality representations (e.g., multiple resolutions to meet the diverse display capabilities of different types of user devices and multiple bitrates to adapt to network characteristics) and is made available on an HTTP server. Each quality representation is organized in smaller and independently decodable files called segments; each segment typically accounting for a few seconds of video. A client desiring to access a video, initially fetches a manifest file from the server, that contains the description of the segments (available quality representations, bitrate of each segment, etc.). Then the client deploys a \emph{rate adaptation algorithm}, that sequentially selects the appropriate bitrate for each segment, given network conditions. Based on the bitrate indication, the client then selects the corresponding quality representation and independently requests and downloads every segment at a finite-sized queue, known as the buffer. By controlling the bitrate for each segment, rate adaptation algorithms aim at matching the video download (or streaming) rate to the channel rate. If the video consumption (or playback) rate is larger than the download rate, the buffer will deplete, eventually leading to a re-buffering event (i.e. playback interruption). In essence, a rate adaptation algorithm is an optimization solution with the objective of maximizing the streamed video bitrate, while at the same time ensuring uninterrupted and stable (i.e. minimal bitrate switches) streaming.

Constant developments in cellular networking technology, such as  4G's \ac{LTE} or the anticipated 5G, are changing the landscape of mobile high-bandwidth multimedia applications, that are becoming fast an integral part of the mobile clients' life. In particular, the demand for \emph{mobile} video streaming has advanced at unprecedented growth rates over the last years and is expected to reach 79\% of the global mobile data traffic by 2022 \cite{ciscomobile}. Nonetheless, cellular networks are typically characterized by throughput variation, caused by radio propagation effects, such as scattering, fast fading, path loss and shadowing or handover events; that occur when a data session is transferred to another cell. Such network conditions pose significant challenges on mobile streaming solutions, where optimal bitrate adaptation over fluctuating wireless channels remains an elusive task. This paper aims to offer a novel perspective on the mobile bitrate adaptation problem, under the scope of online optimization.

As the \ac{DASH} standard does not specify a particular rate adaptation algorithm, a plethora of proposed solutions exists in both scientific literature and actual industry practices. A performance evaluation of recent rate adaptation algorithms in mobile networks \cite{own}, showed that fixed-rule schemes may require parameter tuning according to the considered network or user scenario, and thus cannot generalize well beyond a certain scope of usage. In an effort to overcome this limitation, some algorithms resort to learning techniques or control theoretic approaches to attain optimal bitrate adaptation. However, their practical implementation on mobile devices may be hindered by energy-demanding architectures \cite{8048013} or by the complexity of exploring the complete optimization space \cite{MPC}. In this work we propose a novel rate adaptation algorithm based on \ac{OCO} \cite{Zinkevich:2003:OCP:3041838.3041955}, that is independent of any parameter selection concerning the streaming environment and does not require computationally heavy operations.

\ac{OCO} has emerged as a very effective online learning framework, that is also suitable for mobile deployment, in terms of resource requirements. According to \ac{OCO}, an agent learns to make sequential decisions in order to minimize an adversarial loss function, unknown at decision time. \ac{OCO} is ``model-free", as no statistical assumption is required, while at the same time it provides tractable feasibility and performance guarantees \cite{DBLP:journals/corr/abs-1804-04529}. Having already been proposed for problems with rapidly fluctuating environments, such as cloud resource reservation management \cite{THOR} and dynamic network resource allocation \cite{8027140}, it constitutes an appealing candidate for HAS as well. However, the application of \ac{OCO} in \ac{HAS} is not a straightforward task. Given its discrete decision space (set of available quality representations) and instantaneous state-dependent constraints (finite-sized buffer queue), HAS optimization does not fall directly in the class of \ac{OCO} problems.

This work provides multiple contributions towards formulating the HAS optimization problem under the \ac{OCO} framework. First, we model the adaptive streaming client by a learning agent, whose objective is to maximize the average bitrate of a streaming session, subject to scheduling constraints of the buffer queue. In general, the choice of the objective function is made under the assumption that higher average bitrate typically corresponds to higher video quality, when comparing the same video in the same resolution. The constraints are chosen relative to the buffer queue, since re-buffering events can significantly influence \ac{QoE} \cite{HAS_QoE}. Second, we fulfill the \ac{OCO} requirement that both the set of decisions and constraint functions must be convex by a) allowing the agent to make decisions on the quality representation of each segment, according to a probability distribution for the bitrate and by b) deriving a set of convex constraints associated with the upper and lower bound of the finite buffer. We achieve the latter by making a relaxation to an unbounded buffer that adheres to time-averaging constraints. Third, we model the channel rate evolution by an adversary, that decides the cost of each decision a posteriori. We  eventually solve the HAS optimization problem by proposing \emph{Learn2Adapt} (\emph{L2A}), a novel rate adaptation algorithm based on the \ac{OCO} theory. In our trace-based simulations, our proposed method proves to be \emph{robust}, providing consistently better \ac{QoE}, when evaluated against reference state-of-the-art rate adaptation algorithms in a wide spectrum of possible network and streaming conditions. 

	\section{Related work}\label{related}
Video rate adaptation schemes can be broadly classified according to the module that implements the rate adaptation logic. According to the DASH standard, multimedia delivery requires a server-client architecture, thus the rate adaptation module can be hosted at either of the two components. Server-side rate adaptation methods, require no cooperation from the client and resort to traffic shaping methods at the server-side alone \cite{Akhshabi:2013:STS:2460782.2460786}. Such approaches may produce high overhead on the server and thus make scaling with the number of clients a real challenge. Additionally, network-assisted rate adaptation methods have also been proposed, that allow  HAS clients to take network information into consideration for optimizing the rate adaptation process \cite{Krishnamoorthi:2017:BPB:3083187.3083193,6920044}. Nonetheless, most of the proposed rate adaptation schemes, reside at the client-side, where bitrate decisions are made according to either network or application-level information, a combination of both, or even cross-layer metrics. In the following, we focus our analysis on such client-side approaches, as they are more relevant and thus comparable to the proposed method herein.
  
Primarily, heuristic approaches have been proposed for client-side rate adaptation, that can be mainly classified into three categories according to the considered input. First, throughput-based methods estimate the available channel rate to decide on the bitrate of the streamed video. For instance, Li et al. \cite{Panda} propose, \emph{PANDA}, a rate adaptation module that uses a moving average filter to estimate the available throughput and schedules the download of every segment, in a way that reduces bitrate oscillations, particularly in scenarios with multiple clients. A similar throughput-based strategy, called \emph{FESTIVE} \cite{6704839}, focuses primarily on fairness amongst all clients. Second, buffer-based methods use application-level signals, such as the instantaneous buffer level to perform the adaptation. A notable such method comes from Huang et al. \cite{Netflix}, who propose \emph{BBA}; a mapping between instantaneous buffer values to video bitrate levels. Third, hybrid methods may use a combination of inputs. In that direction, Kim et al. \cite{8411495} propose \emph{XMAS}, a hybrid method that deploys a traffic shaping scheme, based on both throughput estimates and playback buffer levels, while Xie et al. \cite{Xie:2015:PPL:2789168.2790118} propose \emph{piStream}; a physical-layer informed rate adaptation strategy. Lately, new Smartphone devices have enabled the fusion of multiple sensor readings to infer context in the mobile client's environment. To this end, Mekki et al. \cite{IOBBA} solicit incorporating a user's inferred location into the decision process.


Recently, there has been a shift in the scientific literature, in regard to the methods used in the rate adaptation design; primarily towards optimization and control theoretic approaches. Most notably, Spiteri et al. \cite{BOLA} formulate rate adaptation as a utility maximization problem and devise \emph{BOLA}, an online control algorithm, that makes use of the instantaneous buffer occupancy. Also in the direction of control-theoretic schemes, \emph{MPC}, by Yin et al. \cite{MPC}, combines buffer occupancy and throughput predictions for optimal rate adaptation. 

In regard to rate adaptation methods based on optimization and in particular on dynamic programming, Zhou et al. \cite{7393865}, propose \emph{mDASH} and formulate the rate adaptation logic as a \ac{MDP} where the buffer size, bandwidth conditions and bitrate stability are taken as Markov state variables. Similarly, Bokani et al. \cite{7305810} model the rate adaptation logic as an MDP problem as well and incorporate mobility by including vehicular environments. Some of the main drawbacks of MDP-based solutions are computational load and the need to know the statistics of the network and video content in advance. 

Model-free Reinforcement Learning (RL) approaches, such as Q-Learning (QL), have also been investigated for the design of rate adaptation methods. Claeys et al. \cite{4083008} propose a QL-based HAS client, allowing dynamical adjustment of the streaming behavior to the perceived network state. While QL approaches provably converge to the optimal policy, provided that their parameters are chosen correctly, the convergence speed becomes an issue when trying to cope with previously unseen channel or video content patterns.

Lately, Deep-Learning (DL) approaches have also been proposed for rate adaptation, presenting promising merits in both accuracy and convergence of bitrate decisions. \emph{Pensieve} \cite{Pensieve} is a rate adaptation DL framework that does not rely on pre-programmed models or assumptions about the environment, but instead gradually learns the best policy for bitrate decisions through observation and experience. Another DL approach is called \emph{D-DASH} \cite{8048013}, that combines DL and RL mechanisms and achieves a good trade-off between policy optimality and convergence speed during the decision process. Nonetheless, the deployment of DL in mobile devices is typically associated with high computational and energy demands, especially during training phases and thus external (hardware) resources may be required to assist in the rate adaptation process.

Huang et al. \cite{Huang:2019:HEV:3304109.3306219} have recently explored combinatorial optimization for rate adaptation and have proposed \emph{Hindsight}, a near-optimal, linear-time and linear-space greedy algorithm.

During our literature review, we have identified a requirement for a rate adaptation approach, that is not only scalable to the number of clients, but also light in computation and that does not rely on any modelling assumptions. We attempt to fulfill these prerequisites with our novel method proposed in \sref{Algorithm}. A more detailed survey on adaptive streaming solutions can be found in \cite{8424813}.

	\section{System Model}
\label{sec:model}

This section introduces the model for the media content and client operations used in the rest of this work. Moreover, the notation is summarized in \tref{tab:variables}.

\subsection{Media model}
Let us assume that a video sequence of duration $D$ seconds is stored on a server organized in the form of $T=\lceil{}D/V\rceil{}$ segments, each of constant playback duration $V$. Each segment is encoded at $N$ quality representations at increasing \emph{target bitrate} $r \in\{r_1,\dots, r_N\}$. 
For a given quality representation $n \in\{1, \dots, N\}$, the \emph{actual size} of the $t$-th segment ($t \in \{1, \dots, T\}$) -- denoted $S_{t,n}$ and measured in bits -- is a function of the segment content. In the following, we will assume that the server is connected to the client across a channel of rate $C_t$ and thus the $t$-th segment is downloaded across the channel in $\frac{S_{t,n}}{C_t}$ seconds.

\subsection{Client model}
The client issues a request to the server, for the $t$-th segment and then waits for that segment to be fully downloaded before requesting the $(t+1)$-th segment. We refer to the, typically variable, interval between two consecutive requests as a \emph{decision epoch}. Since the content is downloaded in $T$ segments, the total number of decision epochs is $T$ and referred to as the \emph{horizon}. At the beginning of the $t$-th epoch, the client selects the quality representation $x_t \in \mathcal{X}=\{1,\dots N\}$  for segment $t$, corresponding to the bitrate indication $r_{x_t}\in\{r_1,\dots, r_N\}$ of the deployed rate adaptation algorithm. 

Let $B_t$ represent the buffer level at the beginning of the $t$-th epoch, measured in seconds of buffered video at the client. The downloaded segments are stored in a buffer whose size may not exceed an upper bound $B_{max}$, that exists normally due to memory constraints of the mobile device. Upon completely downloading the $t$-th segment, $B_t$ increases by $V$ seconds.

However, due to the concurrent playback of the buffered segments, $B_t$ will also decrease by the amount of time required to download the $t$-th segment, which is equal to $\frac{S_{t,x_t}}{C_t}$ seconds (as long as $B_t>0$). So, the buffer level evolves between two consecutive epochs according to:
\begin{equation}\label{eq:buffer_evolution}
B_{t+1}=\left[ B_{t}-\frac{S_{t,x_t}}{C_t} \right]^++V - \Delta_t,
\end{equation}
\noindent
where $[x]^+ \triangleq\max(0,x)$. A delay $\Delta_t=\left[ B_{t}-\frac{S_{t,x_t}}{C_t}+V-B_{max}\right]^+$ is  introduced to account for the upper bound $B_{max}$ of the buffer size. In other words, if $B_{t}-\frac{S_{t,x_t}}{C_t}+V < B_{max}$, the $(t+1)$-th segment is requested immediately and $\Delta_t = 0$. Otherwise, the request for the $(t+1)$-th segment is delayed by $\Delta_t$ seconds, to allow the buffer to drop to $B_{max}$. This delay protects against \emph{buffer overflow} incidents, which occur when the buffer surpasses $B_{max}$ and creates the characteristic bursty traffic of \ac{HAS} \cite{TPclassification}.
A \emph{buffer underflow} occurs when the instantaneous buffer level drops below zero, causing a \emph{stall} in the video playback, an event that significantly degrades the \ac{QoE} \cite{HAS_QoE}.

In the next section we provide a framework which allows us to design a learning algorithm, that provably optimizes video quality subject to keeping the buffer asymptotically away from the two limits.

 \begin{table}[t]\caption{Notations}
\resizebox{\columnwidth}{!}{%
\centering
\begin{tabular}{c|c|c}
   Notation & Definition & Units \\  
\hline
 $D$ & Video content total duration & seconds \\ 
 $V$& Segment duration & seconds\\
 $T$& Streaming horizon & segments \\
 $N$ & Quality representations & scalar\\
 $x_t$      & Selected quality representation for segment $t$ & scalar\\
 $r_{x_t}$  & Bitrate corresponding to quality  $x_t$& kbps\\
 $S_{t,n}$  & File size of segment $t$ in $n$-th quality & kbits \\
 $\omega_t$     & Decision distribution & probability vector\\
 $\omega^*$     & Benchmark distribution & probability vector \\
 $Q$        & Virtual queue  & scalar\\
 $V_L$      & Cautiousness parameter & scalar\\
 $\alpha$   & Step-size & scalar\\
 $\beta$    & Target switching rate & switches per epoch\\
 $\gamma$   & Switch counter & scalar\\
 $C_t$      & Channel rate at epoch $t$ & kbps\\
 $B_t$      & Buffer level at epoch $t$ & seconds\\
 $B_{max}$  & Maximum buffer level & seconds\\
 $\Delta_t$ & Buffer delay & seconds\\
\end{tabular}}
\label{tab:variables}
\end{table}

	\section{Adaptive streaming problem formulation}\label{Algorithm}
This section  provides an  algorithmic solution based on the theory of \ac{OCO} \cite{Zinkevich:2003:OCP:3041838.3041955}. In order to cast the video streaming optimization problem as an \ac{OCO} with budget constraints problem, we first propose a relaxation on the finite buffer queue and then we modify the formulation to convexify the decision space. In the following, we present our online-learning algorithm \emph{Learn2Adapt} (\emph{L2A}), based on gradient descent and we provide theoretical guarantees for its performance. 

\subsection{\ac{OCO} formulation}
We formulate the rate adaptation problem as a \emph{constrained \ac{OCO}} problem, where the goal is to minimize the cumulative losses $\sum_{t=1}^Tf_t(x_t)$ (referring to the average bitrate of the downloaded segments) while keeping the cumulative constraint functions $\sum_{t=1}^Tg^i_t(x_t)$, $~\forall i=1,2$, negative (referring to buffer underflow and overflow); see also the relevant literature \cite{COLD, 2017arXiv170204783N}.
In the \ac{OCO} framework, functions $f_t,g^i_t$ $\forall i=1,2$ are chosen by an \emph{adversary} and are unknown at decision time. We will relate these functions to the random evolution of the channel rate $C_t$, which in nature is not adversarial. Nevertheless, the adversarial setting is more general and includes any -- potentially time-varying -- distribution of $C_t$, which in turn bestows on our algorithm superior \emph{robustness}.  Next, we explain how these functions are used in our system.

Recall the set of quality representations $\mathcal{X}$, and let $x_t\in\mathcal{X}$ be the decision for the quality representation of the segment to be downloaded in epoch $t$. Consider the following functions:
\begin{align}
\tilde{f}_t(x_t)&\triangleq-r_{x_t} \label{eq:d_f}\\
\tilde{g}^1_t(x_t)&\triangleq\frac{S_{t,x_t}}{C_t}-V, \label{eq:d_g1}\\
\tilde{g}^2_t(x_t)&\triangleq V-\frac{S_{t,x_t}}{C_t}-\frac{B_{max}}{T}, \label{eq:d_g2}
\end{align}
where ~\eqref{eq:d_f} captures the utility (higher bitrate yields smaller losses). \eqref{eq:d_g1}-\eqref{eq:d_g2}  express the buffer displacement, which will be used to model the buffer underflow and overflow constraints, respectively. A high quality representation $x_t$ combined with a low channel rate $C_t$ will prolong download time $\frac{S_{t,x_t}}{C_t}$, which will result in high  buffer consumption. Since $C_t$ is unknown at decision time of $x_t$, it is impossible to know the values of $\tilde{g}^i_t(x_t),~\forall i=1,2$. Our approach therefore, is to learn the best $x_t$ based on our estimation of $\tilde{g}^i_t(x_t),~\forall i=1,2$.

To cast the above problem as \ac{OCO} with budget constraints, we propose the following steps:
\begin{itemize}
\item First, we provide a relaxation to the hard constraints of the buffer model.
\item Second, we convexify the decision set by randomization. We associate a probability to each decision, and we learn the optimal probability distribution for deciding the quality representation to download at each epoch.
\end{itemize}


\subsection{Buffer constraints}
Here we explain how we use the cumulative constraint functions $\sum_{t=1}^T\tilde{g}^i_t(x_t),~\forall i=1,2$ to model buffer underflow and overflow, respectively. The buffer evolves according to \eqref{eq:buffer_evolution} and ensuring $0 \leq B_{t}\leq B_{max},~\forall t$,  involves in principle a very complicated control problem, which in the presence of unknown adversarial $C_t$ is exacerbated.

To avoid computationally heavy approaches and to arrive at a simple (yet robust) solution, we thus seek an alternative approach. In that direction, we treat the buffer as an infinite queue, with the simpler (compared to ~\eqref{eq:buffer_evolution})  update rule: $B_{t+1}=B_{t}+V-S_{t,x_t}/C_t$, where now no additional delay $\Delta_t$ is ever imposed on the system. By this, we allow instantaneous violation of the budget, but we utilize a penalty which aims to maintain the buffer on the $[0,B_{max}]$ range on average. In particular, using  \eqref{eq:d_g1}-\eqref{eq:d_g2}, we capture in $\tilde{g}^i_t(x_t),~\forall i=1,2$ the instantaneous buffer displacement on both directions (measured in seconds) and by requiring the cumulative constraint $\sum_{t=1}^T \tilde{g}^i_t(x_t) \leq 0,~\forall i=1,2$, we ensure that on average $B_t$ remains in the non-negative regime below $B_{max}$.  A benefit is that these constraints are in the realm of \ac{OCO} theory, and therefore allow us to design a simple learning algorithm that provably satisfies them.  Overall, our approach here is to apply a loosely coupled control to the buffer constraints, by tolerating instantaneous violations and ensuring that in the long-term only a few are experienced.

\subsection{Convexification}
To obtain a convex decision set, we use a convexification method based on randomization of the decision process \cite{Shalev-Shwartz:2012:OLO:2185819.2185820}. Consider the probability simplex:
\begin{equation*}
{\Omega}=\{\bm{\omega} \in \mathbb{R}^N: \bm{\omega} \geq 0 \wedge \|{\bm{\omega}}\|_1=1 \},
\end{equation*}
where $\omega_n=\mathbb{P}(x=n)$ denotes the probability that we decide $x=n\in\{1,\dots,N\}$ and $\Omega$ is a convex set. Thus, instead of learning directly the decision $x_t$, we learn the optimal probability $\bm{\omega}_t=(\omega_{t,n})_{n=1,\dots,N}$ of picking  $x_t$ from the integer set $\mathcal{X}$.
Given a decision $\bm{\omega}_t$, the actual quality representation will be chosen according to the expectation of the corresponding utility, i.e. $x_t \in \arg\min_{x\in \mathcal{X}} |r_x-\sum_{n=1}^N \omega_{t,n} r_n|$. The functions of interest become now random processes and we must appropriately modify them by taking expectations with respect only to $\bm{\omega}_t$ and not to the randomness of $C_t$:

\resizebox{0.88\columnwidth}{!}{
\begin{minipage}{\linewidth}
\begin{align}
f_t(\bm{\omega}_t)&\triangleq-\mathbb{E}[r_{x_t}] = -\sum_{n=1}^N \omega_{t,n}r_n \label{eq:c_f}\\ 
g^1_t(\bm{\omega}_t)&\triangleq\mathbb{E}\left[\frac{S_{t,x_t}}{C_t}-V\right]=\frac{\sum_{n=1}^N \omega_{t,n}S_{t,n}}{C_t}-V \label{eq:c_g1}\\
g^2_t(\bm{\omega}_t)&\triangleq\mathbb{E}\left[V-\frac{S_{t,x_t}}{C_t}-\frac{B_{max}}{T}\right]=V-\frac{\sum_{n=1}^N \omega_{t,n}S_{t,n}}{C_t}-\frac{B_{max}}{T}. \label{eq:c_g2}
\end{align}
\end{minipage}}

Given the loss function and constraints above, we formulate the constrained \ac{OCO} problem, that we solve in \sref{solution}:
\begin{align*}
\label{eqn:opt_}
 \min_{\bm{\omega}\in\Omega} & \sum_{t=1}^T f_t(\bm{\omega}) \quad
\text{s.t. }\quad \sum_{t=1}^{T} g^i_t(\bm{\omega}) \leq 0 \quad \forall i=1,2.
\end{align*}

The following are true for functions \eqref{eq:c_f}-\eqref{eq:c_g2} and our surrogate convex problem:
\begin{itemize}
	\item The diameter of $\Omega$, defined as the largest Euclidean distance between any two vectors, is $\sqrt{N}$. 
	\item Functions $f_t$ and $g^i_t,~\forall i=1,2$ are smooth, bounded and have bounded gradients. Specifically, $\forall t,\bm{\omega},i=1,2$:
\resizebox{0.9\columnwidth}{!}{
\begin{minipage}{\linewidth}
\begin{align*}
	& |f_t(\bm{\omega})| \leq r_N, \\
	&|g^1_t(\bm{\omega})|\leq \max\left\{\left|\frac{S_{\min}}{C_{\max}}-V\right|,\left|\frac{S_{\max}}{C_{\min}}-V\right|\right\}, \\
	& |g^2_t(\bm{\omega})|\leq  \max\left\{\left|V-\frac{S_{\min}}{C_{\max}}-\frac{B_{max}}{T}\right|,\left|V-\frac{S_{\max}}{C_{\min}}-\frac{B_{max}}{T}\right|\right\},\\
	& \|\nabla f_t(\bm{\omega})\|\leq \sqrt{\sum_{n=1}^Nr_n^2},  \quad \|\nabla g^i_t(\bm{\omega})\|\leq \sqrt{\sum_{n=1}^N\left(\frac{S_{t,n}}{C_{\min}}\right)^2}, \\
\end{align*}
\end{minipage}}
	where $C_t\in [C_{\min},C_{\max}]$, $S_{t,j}\in [S_{\min},S_{\max}]$
	and $\nabla f_t(\bm{\omega}),\nabla g^i_t(\bm{\omega}),~\forall i=1,2$ denote the gradients.
\end{itemize}

\subsection{Regret metric}
At every decision epoch $t=1,2,\dots T$  the following events occur in succession:
\begin{enumerate}[(a)]
\item the  agent computes $\bm{\omega}_t\in \Omega$ according to an algorithm,
\item the  agent chooses $x_t \in \arg\min_{x\in \mathcal{X}} |r_x-\sum_{n=1}^N \omega_{t,n} r_n|$,
\item an adversary decides $C_t$, and the loss function  $\tilde{f}_t(\bm{\omega}_t)$ and the constraint functions $\tilde{g}^i_t(\bm{\omega}_t),~\forall i=1,2$ are determined using ~\eqref{eq:d_f}-\eqref{eq:d_g2}, and then used to measure the actual loss and buffer displacement,
\item the following forms of feedback are provided to the agent: (i) the value of $C_t$, (ii) the functions $f_t,g^i_t,~\forall i=1,2$, (iii) the gradients $\nabla f_t(\bm{\omega}_t), \nabla g^i_t(\bm{\omega}_t),~\forall i=1,2$.
\end{enumerate}

The feedback above is used by the agent to eventually determine the gradient vectors $\nabla f_{t+1}(\bm{\omega}_{t+1}),\nabla g^i_{t+1}(\bm{\omega}_{t+1}),~\forall i=1,2$.
We now define the performance metric in our problem which consists of two parts: the \emph{regret} of an algorithm  and the $i$-th \emph{constraint residual}, defined as:
\[
R_T = \sum_{t=1}^T f_t(\bm{\omega}_t)- \sum_{t=1}^T f_t(\bm{\omega}^*) \quad \text{and} \quad V^i_T = \sum_{t=1}^Tg^i_t(\bm{\omega}_t),
\]
respectively. Here $\bm{\omega}^*\in \Omega$ is a benchmark distribution,  that minimizes the losses in hindsight, with knowledge of the functions $f_t,g^i_t,~\forall i=1,2$. This benchmark satisfies the cumulative constraints every $K$:
\begin{align*}
\label{eqn:opt_final}
\bm{\omega}^{*}\in \arg\min_{\bm{\omega}\in\Omega} & \sum_{t=1}^T f_t(\bm{\omega}) \\
\text{s.t. } & \sum_{t=k}^{K+k-1} g^i_t(\bm{\omega}) \leq 0,\\ &
~~\forall k=1,\dots,T-K+1, ~~\text{and} \quad \forall i=1,2.
\end{align*}

This benchmark is first explained in \cite{COLD}, where the authors prove that for any $K=o(T)$,  a smart agent can learn to have no regret, while satisfying the adversarial constraints. In our case, picking $K=T^{1-\epsilon}$, for small $\epsilon>0$, gives the best approximation of our algorithms' performance, allowing maximum freedom to the competing benchmark.
If an algorithm achieves both $o(T)$ regret and $o(T)$ constraint residual, then it follows that as $T\to\infty$ we have (i) $R_T/T\to 0$, hence our algorithm has the same losses with (or ``learns'') the benchmark action, and (ii) $V^i_T/T\to 0,~\forall i=1,2$, hence our algorithm ensures the average constraint. Since the benchmark action is the best \emph{a posteriori} action, taken with knowledge of all the revealed values of $C_t$, learning it is both remarkable and very useful.


\section{\ac{OCO} solution}
\label{solution}

In this section, we propose a ``no regret'' algorithm to solve the constrained \ac{OCO} problem defined in the previous section. We first provide the intuition behind the algorithm design and the introduction of a switching budget, that allows the control of the switching frequency for our algorithm. We then detail the proposed algorithm, and finally provide some performance bounds.

\subsection{Learn to Adapt (L2A) algorithm}

As a general note, a main challenge in such problems is that the constraints $g^i_t(\bm{\omega}_t),~\forall i=1,2$ are not known when the decision of $\bm{\omega}_t$ is taken.
The \ac{OCO} approach to this issue is to predict such functions using a first order Taylor expansion of $g^i_t(\bm{\omega}_t),~\forall i=1,2$ around $\bm{\bm{\omega}}_{t-1}$ evaluated at $\bm{\bm{\omega}}_{t}$
\cite{Zinkevich:2003:OCP:3041838.3041955}: 
\begin{equation}\label{prediction}
\hat{g}^i_t(\bm{\omega}_t)\triangleq g^i_{t-1}(\bm{\omega}_{t-1})+\langle \nabla g^i_{t-1}(\bm{\omega}_{t-1}),\bm{\omega}_t-\bm{\omega}_{t-1}\rangle,~ \forall i=1,2.
\end{equation}
We recall that in (\ref{prediction}), only $\bm{\omega}_t$ is unknown at $t$, whereas $\bm{\omega}_{t-1}$, $\nabla g^i_{t-1}(\bm{\omega}_{t-1})$ and $g^i_{t-1}(\bm{\omega}_{t-1}),~ \forall i=1,2$ are known via the obtained feedback.

Contrary to the standard (unconstrained) online gradient \cite{Zinkevich:2003:OCP:3041838.3041955}, our algorithm must combine the objective and the constraint functions. To this end, consider the regularized Lagrangian:
\[
L_t(\bm{\omega},\bm{Q}(t))=\sum_{i=1}^2Q_i(t)\hat{g}^i_t(\bm{\omega})+ V_L\hat{f}_t(\bm{\omega})+\alpha||\bm{\omega}_t-\bm{\omega}_{t-1}||^2,
\]
where $Q_i(t)$ is the Lagrange multiplier, $\hat{g}^i_t(\bm{\omega})$ is the prediction of the constraint function $g^i_t(\bm{\omega})$ from ~\eqref{prediction}, $V_L$ is a \emph{cautiousness} parameter that controls the trade-off between regret and constraint residual, $\hat{f}_t(\bm{\omega})$ applies \eqref{prediction} to $f_t$, $\alpha$ is the step-size and  $||\bm{\omega}_t-\bm{\omega}_{t-1}||^2$ is a regularization term that smooths the decisions.
Parameters $V_L$ and $\alpha$ are tuned for convergence and their choices are given below.
We mention here, that the Lagrange multiplier $Q_i(t),~\forall i=1,2$ is updated in a \emph{dual ascent} approach, by accumulating the constraint deviations:
\[
Q_i(t+1)=[Q_i(t)+\hat{g}^i_t(\bm{\omega}_t)]^+,~ \forall i=1,2.
\]

We  further compound the online optimization problem by introducing a switching budget. Let $\beta \in(0,1]$ be the maximum allowed reconfiguration frequency measured in quality switches per epoch. The goal is to limit the number of changes within the horizon to at most $\beta T$\footnote{While $\beta$  may allow a switch at a given epoch $t$, $\bm{\omega}_t=\bm{\omega}_{t-1}$ is still a valid decision.}. This is a valuable property that allows stability control for the following algorithm (Algorithm \ref{alg:l2a}), that takes a step in the direction of the sub-gradient of the regularized Lagrangian.

\begin{algorithm}[H]
	\caption{Learn2Adapt (L2A)}
	\begin{algorithmic}[1]
		\renewcommand{\algorithmicrequire}{\textbf{Initialize:}}
		\renewcommand{\algorithmicensure}{\textbf{Parameters:}}
		\REQUIRE  $\bm{Q}(1)=0$, $\bm{\omega}_0 \in S$, $t'=1$
		\ENSURE  cautiousness parameter $V_L$, step size $\alpha$, maximum allowed switch rate $\beta$, switch counter $\gamma=0$
		\\
		\FOR {all $t \in \{1,2,\dots,T\} $}
		\IF{$\frac{\gamma}{t}\leq\beta$}
		\STATE{	\resizebox{0.85\hsize}{!}{$\bm{\bm{\omega}}_t=\vctproj[\Omega]{\left[\bm{\bm{\omega}}_{t-1}-\frac{\sum_{j=t'}^{t} \{V_L \nabla f_{j-1}(\bm{\bm{\omega}}_{j-1})+\sum _{i=1}^2Q_i(j)\nabla g^i_{j-1}(\bm{\bm{\omega}}_{j-1})\}}{2\alpha}\right]}$}}
		\STATE{$t'=t+1$}
		\STATE{$\gamma++$}
		\ELSE
		\STATE{$\bm{\omega}_{t}=\bm{\omega}_{t-1}$}
		\ENDIF
		\STATE {$Q_i(t+1)=[Q_i(t)+\hat{g}^i_t(\bm{\bm{\omega}}_t)]^+,~\forall i=1,2$}
		\ENDFOR
	\end{algorithmic} 
	\label{alg:l2a}
\end{algorithm}


Here $\vctproj[\Omega]{[\cdot]}$ denotes the Euclidean projection on set $\Omega$.

\subsection{Performance guarantees}

The main contribution of this work is the formulation the rate adaptation problem in the constrained \ac{OCO} framework and the proposal of \emph{Learn2Adapt} (\emph{L2A}). In the following we invoke the theorem from \cite{COLD}, to provide theoretical performance guarantees for the \emph{Learn2Adapt} algorithm. We note here that although the following theoretical guarantees are derived for $\beta=1$, in the numerical evaluation of \sref{Numerical_evaluation} we provide evidence that \emph{L2A} performs well even for $\beta<1$.

\begin{thm}[From \cite{COLD}]\label{thm:theorem1}
For $\beta=1$, choose small $\epsilon>0$, fix $K=o(T^{1-\epsilon})$, $V_L = T^{1-\epsilon/2}$, and $\alpha = V_L\sqrt{T}$. Then, 
the \emph{Learn2Adapt (L2A)} algorithm guarantees:
\begin{align*}
R_T  = O(T^{1-\epsilon/2}),\quad\quad V^i_T &=O(T^{1-\epsilon/4}),~ \forall i=1,2.
\end{align*}
\end{thm}

Effectively, this means that over time our algorithm learns the best a-posteriori distribution $\bm{\omega}^*$, which neatly satisfies the average constraints and minimizes the cumulative quality losses.
We experimentally verify below that the corresponding choices $x_t$ made by sampling this distribution have extremely well performing properties for video streaming adaptation.

	\section{Experimental Evaluation}\label{Numerical_evaluation}
In this section we evaluate the performance of our proposed rate adaptation algorithm  against two reference rate adaptation schemes, by experimenting with real mobile network traces and video sequences, for two separate streaming applications (\ac{VoD} and live streaming) and under five video streaming performance metrics.
\subsection{Experimental setup}\label{simulation_parameters}
\paragraph{Network scenarios}
In our evaluation, we use \emph{real} cellular network traces and in particular a data-set that includes 4G channel measurements for various mobility scenarios \cite{Raca:2018:BTL:3204949.3208123}. For our experiments, we have selected the \emph{static}, \emph{pedestrian} and \emph{car} scenarios (operator A therein), as realistic cases for no, low and high mobility, respectively. The \{\emph{static}, \emph{pedestrian}, \emph{car}\} scenario consist of \{$12, 26, 41$\} traces with an average measurement duration of \{$17, 18, 23$\} min, respectively. Cellular networks present significant challenges to the rate adaptation process, as they are typically characterized by rapid throughput fluctuation and short service outages; that may be caused by radio propagation effects, low-coverage areas or handover events. While these characteristics are realistically depicted in the selected traces of \cite{Raca:2018:BTL:3204949.3208123}, taking a step further in our evaluation, we have designed a \emph{synthetic} scenario consisting of 20 traces, that is characterized by abrupt and steep channel rate transitions. This so-called \emph{markovian} scenario emulates two channel levels (states) $\{0.75, 23.0\}$ Mbps with a $0.05$ state transition probability and is complementary to the real traces; to present the rate adaptation algorithms with an additional, even more demanding network scenario.
\paragraph{Video parameters} 
In \cite{Zabrovskiy:2018:MDD:3204949.3208140} video sequences are encoded at multiple bitrates in conditions typical of \ac{OTT} video delivery.
We used 3 sequences: \emph{BBB}, \emph{TOS} and \emph{Sintel}, encoded in the H.264/AVC standard, at target bitrates $\{0.37, 0.75, 1.5, 3.0, 5.8, 12.0, 17.0, 20.0\}$ Mbps, corresponding to resolutions in $\{384\times216, 640\times360, 1024\times576, 1280\times720, 1920\times1080, 3840\times2160\}$, and organized in DASH segments with duration $V=2$s.
\paragraph{Streaming scenarios}
In our experiments, we consider a \ac{VoD} streaming scenario and a live streaming scenario. For the \ac{VoD} scenario, we considered a maximum buffer value of $B_{max}=120$s (60 segments). For the live streaming scenario, we reduced the maximum buffer value to $B_{max}=20$s (10 segments), according to the tighter latency requirements. All the figures below concern the case of \ac{VoD}, while the results for the live streaming scenario are presented in \tref{tab:live}.
\paragraph{Algorithms} 
We compare our method \emph{L2A}, for $\beta=0.3$ and $\beta=1$, against \emph{RB}, a throughput-based method and \emph{BB}, a buffer-based method, following the design principles and parameters selection found in \cite{Panda} and \cite{BOLA}, respectively. 

\emph{RB} \cite[Section VI]{Panda} is a throughput-based rate adaptation scheme based on a four-step adaptation model, where initially the available network bandwidth is estimated using a proactive probing mechanism, that is designed to minimize bitrate oscillations. Then, the throughput estimates are smoothed using noise-filters to avoid errors due to throughput variation and each segment download is scheduled according to inter-request times, that would drive the buffer to the maximum level.

\emph{BB} {BOLA} is a buffer-based rate adaptation algorithm that uses Lyapunov optimization in order to indicate the bitrate of each segment. Practically, the algorithm is designed to maximize a joint utility function that rewards increases in the average bitrate and penalizes stalls. The implemented variant, called \emph{BOLA-O}, mitigates bitrate oscillations by a form of bitrate capping when switching to higher bitrates.

These rate adaptation methods are widely used in research, each amongst the best performing methods of their class \cite{own}. Regarding our method \emph{L2A}, the presented results consider a cautiousness parameter of $V_L=T^{0.9}$ and step size of $\alpha=V_L\sqrt{T}$. We note here that according to \ac{DASH}, in case of a stall, $\tau$ segments must be downloaded in order for the play-out to resume. For all algorithms we considered  $\tau=2$.
\paragraph{Video streaming performance metrics} 
 \begin{table*}[!t]\caption{Video streaming performance metrics}
\centering
\begin{tabular}{c|c|c}
   \textbf{Metric name} &  Element evaluated & Metric \\  
\hline &&\\
\textbf{Average bitrate} & Average video bitrate    &  $\frac{\bar{r}}{\max_{ m \in \mathcal{M}}{\bar{r}^{m}}}$ \\

\textbf{Stability} & Bitrate switching frequency& $1-\frac{\sum_{t=2}^T\left(\bm{\mathbbm{1}}_{\{r_{x_t} \neq r_{x_{t-1}}\}}\right)}{T-1}$ \\

\textbf{Smoothness} & Adaptation amplitude & $1-\frac{\sum_{t=2}^T|r_{x_t}-r_{x_{t-1}}|}{(r_N-r_1)(T-1)}$\\

\textbf{Consistency}& Stall duration & $1-\frac{\sum_{t=1}^T \bm{\mathbbm{1}}_{\left\{B_{t-1}<\frac{S_{t,x_t}}{C_t}\right\}}\left(-B_{t-1}+\sum_{k=0}^{\tau-1}\frac{S_{t+k,x_{t+k}}}{C_{t+k}}\right)}{D}$\\  

\textbf{Continuity}& Frequency of stalls & $1-\frac{\sum_{t=1}^T \bm{\mathbbm{1}}_{\left\{B_{t-1}<\frac{S_{t,x_t}}{C_t}\right\}}}{\lceil{}\frac{T}{\tau}\rceil}$ \\  \end{tabular}
\label{tab:table_metrics}
\end{table*}
We evaluate the performance of our proposed method based on the  video streaming performance metrics presented in \tref{tab:table_metrics}. \emph{Average bitrate} models the average bitrate $\bar{r}=\frac{\sum_{t=1}^Tr_{x_t}}{T}$ of the received video segments in a session, normalized over the maximum average bitrate $\max_{m\in\mathcal{M}}\bar{r}^m$ obtained for that session by any adaptation method $m\in\mathcal{M}$, where $\mathcal{M}$ is the set of all evaluated methods. Streaming \emph{stability} models the frequency of bitrate switching, while streaming \emph{smoothness} is associated with the amplitude of the bitrate switches, i.e the absolute bitrate difference between sequential segments. Both \emph{stability}  and \emph{smoothness} are normalized over the maximum attainable value for each respective metric, while $\bm{\mathbbm{1}}_{\{\mathcal{Y}\}}$ is an indicator vector; with ones at the positions that condition $\mathcal{Y}$ is true, and zeros otherwise. Additionally, we propose two metrics associated with a) the frequency of stalls and b) their severity (duration). With streaming \emph{consistency} we measure the percentage of the user's allocated time-budget (typically equal to the video length $D$) that was spent actually consuming video content (as opposed to stalling), while streaming \emph{continuity} expresses the percentage of segments that were downloaded while play-out remained uninterrupted, assuming $B_0=0$.
\subsection{Results}\label{results}

\begin{figure*}[!t]
	\centering
%
%
\definecolor{mycolor5}{RGB}{171,217,233}
\definecolor{L2A1}{RGB}{44,123,182}

\definecolor{L2A2}{RGB}{230,97,1}

\definecolor{RB}{RGB}{253,184,99}
\definecolor{BB}{rgb}{0.24, 0.71, 0.54}
\begin{tikzpicture}

\begin{axis}[%
hide axis,
xmin=10,
xmax=50,
ymin=0,
ymax=0.4,
legend style={at={(0.5,1.05)}, anchor=north, legend columns=4, legend cell align=left, align=left, draw=white!15!black}
]

\addlegendimage{color=BB, line width=2.0pt}
\addlegendentry{BB};
\addlegendimage{color=RB, line width=2.0pt}
\addlegendentry{RB};
\addlegendimage{color=L2A1, line width=2.0pt}
\addlegendentry{L2A ($\beta=0.3$)};
\addlegendimage{color=L2A2, line width=2.0pt}
\addlegendentry{L2A ($\beta=1$)};

\end{axis}
\end{tikzpicture}%
	
	\begin{tabular}{C{.49\textwidth}C{.49\textwidth}}
		\subfigure [Static user] {
			\resizebox{0.98\linewidth}{!}{%
%
%
\definecolor{mycolor5}{RGB}{171,217,233}
\definecolor{L2A1}{RGB}{44,123,182}

\definecolor{L2A2}{RGB}{230,97,1}

\definecolor{RB}{RGB}{253,184,99}
\definecolor{BB}{rgb}{0.24, 0.71, 0.54}
\begin{tikzpicture}

\begin{axis}[%
width=4.521in,
height=3.566in,
at={(0.758in,0.481in)},
scale only axis,
xmin=-1.39060329693044,
xmax=1.39060329693044,
ymin=-1.0967822777403,
ymax=1.0967822777403,
axis line style={draw=none},
ticks=none,
title style={font=\bfseries},
legend style={at={(0.5,1.05)}, anchor=north, legend columns=4, legend cell align=left, align=left, draw=white!15!black}
]
\addplot [color=BB, line width=2.5pt]
  table[row sep=crcr]{%
0.524892558816627	0.722452628145057\\
0.657721066647685	-0.213706529181149\\
1.11022302462516e-16	-0.72054456003972\\
-0.927954768958132	-0.301510781647753\\
-0.524272579844317	0.721599300296497\\
0.524892558816626	0.722452628145057\\
};
\addlegendentry{BB}
\addplot [color=RB, line width=2.5pt]
  table[row sep=crcr]{%
0.327676945054557	0.451008622928319\\
0.910972933680685	-0.295993048887928\\
1.11022302462516e-16	-0.891386207749191\\
-0.940490114881122	-0.305583762437226\\
-0.522766650454295	0.719526566310565\\
0.327676945054557	0.451008622928319\\
};
\addlegendentry{RB}

\addplot [color=L2A1, line width=2.5pt]
  table[row sep=crcr]{%
0.534907859620369	0.736237507099409\\
0.878194910147752	-0.285342823438492\\
1.11022302462516e-16	-0.873684400501662\\
-0.922118472248343	-0.299614453893683\\
-0.548548648865743	0.755012442797698\\
0.534907859620369	0.736237507099409\\
};

\addlegendentry{L2A ($\beta=0.2$)}

\addplot [color=L2A2, line width=2.5pt]
  table[row sep=crcr]{%
0.559773898210238	0.791659939041436\\
0.761197835987872	-0.247328169642325\\
1.11022302462516e-16	-0.831554053981203\\
-0.921586798903884	-0.299441702752107\\
-0.55954670614783	0.770149970001069\\
0.559773898210238	0.791659939041436\\
};
\addlegendentry{L2A ($\beta=1$)}

\node[right, align=left]
at (axis cs:0,0) {0};
\node[right, align=left]
at (axis cs:0.079,0.233) {0.25};
\node[right, align=left]
at (axis cs:0.159,0.466) {0.5};
\node[right, align=left]
at (axis cs:0.238,0.699) {0.75};
\node[right, align=left]
at (axis cs:0.430,0.850) {1};
\addplot [color=black, dotted, forget plot]
  table[row sep=crcr]{%
0	0\\
0.58283685864599	0.802206114824554\\
};
\addplot [color=black, dotted, forget plot]
  table[row sep=crcr]{%
0	-0\\
0.94304984718543	-0.306415469879979\\
};
\addplot [color=black, dotted, forget plot]
  table[row sep=crcr]{%
0	-0\\
1.11022302462516e-16	-0.99158128988915\\
};
\addplot [color=black, dotted, forget plot]
  table[row sep=crcr]{%
-0	-0\\
-0.94304984718543	-0.306415469879979\\
};
\addplot [color=black, dotted, forget plot]
  table[row sep=crcr]{%
-0	0\\
-0.58283685864599	0.802206114824554\\
};
\addplot [color=black, dotted, forget plot]
  table[row sep=crcr]{%
0	0\\
0.58283685864599	0.802206114824554\\
};
\addplot [color=black, dotted, forget plot]
  table[row sep=crcr]{%
0	0\\
0	0\\
};
\addplot [color=black, dotted, forget plot]
  table[row sep=crcr]{%
0.145709214661498	0.200551528706138\\
0.235762461796358	-0.0766038674699946\\
2.77555756156289e-17	-0.247895322472288\\
-0.235762461796357	-0.0766038674699947\\
-0.145709214661498	0.200551528706138\\
0.145709214661497	0.200551528706138\\
};
\addplot [color=black, dotted, forget plot]
  table[row sep=crcr]{%
0.291418429322995	0.401103057412277\\
0.471524923592715	-0.153207734939989\\
5.55111512312578e-17	-0.495790644944575\\
-0.471524923592715	-0.153207734939989\\
-0.291418429322995	0.401103057412277\\
0.291418429322995	0.401103057412277\\
};
\addplot [color=black, dotted, forget plot]
  table[row sep=crcr]{%
0.437127643984493	0.601654586118415\\
0.707287385389073	-0.229811602409984\\
1.11022302462516e-16	-0.743685967416863\\
-0.707287385389072	-0.229811602409984\\
-0.437127643984493	0.601654586118415\\
0.437127643984492	0.601654586118415\\
};
\addplot [color=black, dotted, forget plot]
  table[row sep=crcr]{%
0.58283685864599	0.802206114824554\\
0.94304984718543	-0.306415469879979\\
1.11022302462516e-16	-0.99158128988915\\
-0.94304984718543	-0.306415469879979\\
-0.58283685864599	0.802206114824554\\
0.58283685864599	0.802206114824554\\
};
\node[align=center,rotate=-45]
at (axis cs:0.75,0.8) {\large \textbf{Average bitrate}};
\node[align=center,,rotate=45]
at (axis cs:1,-0.4) {\large \textbf{Stability}};
\node[align=center]
at (axis cs:0,-1) {\large \textbf{Smoothness}};
\node[align=center,rotate=-45]
at (axis cs:-1,-0.4) {\large \textbf{Consistency}};
\node[align=center,rotate=45]
at (axis cs:-0.75,0.8) {\large \textbf{Continuity}};
\legend{}
\end{axis}
\end{tikzpicture}%
				\label{fig:static_radar}      
			}
		}&\subfigure [Pedestrian user] {
			\resizebox{0.98\linewidth}{!}{%
%
%
\definecolor{mycolor5}{RGB}{171,217,233}
\definecolor{L2A1}{RGB}{44,123,182}

\definecolor{L2A2}{RGB}{230,97,1}

\definecolor{RB}{RGB}{253,184,99}
\definecolor{BB}{rgb}{0.24, 0.71, 0.54}
\begin{tikzpicture}

\begin{axis}[%
width=4.521in,
height=3.566in,
at={(0.758in,0.481in)},
scale only axis,
xmin=-1.39060329693044,
xmax=1.39060329693044,
ymin=-1.0967822777403,
ymax=1.0967822777403,
axis line style={draw=none},
ticks=none,
title style={font=\bfseries},
legend style={at={(0.5,1.05)}, anchor=north, legend columns=4, legend cell align=left, align=left, draw=white!15!black}
]
\addplot [color=BB,  line width=2.0pt]
  table[row sep=crcr]{%
0.502983292266075	0.692297109774094\\
0.700151955706077	-0.227493160764894\\
1.11022302462516e-16	-0.953518340927663\\
-0.919048046083737	-0.298616811956974\\
-0.519035432176821	0.714390984932119\\
0.502983292266075	0.692297109774094\\
};
\addlegendentry{BB}
\addplot [color=RB,, line width=2.0pt]
  table[row sep=crcr]{%
0.315640824008286	0.434442323527629\\
0.88461406389851	-0.287428532925261\\
1.11022302462516e-16	-0.972987103337034\\
-0.902894959729437	-0.293368356045511\\
-0.396214929030236	0.545343064937985\\
0.315640824008286	0.434442323527629\\
};
\addlegendentry{RB}

\addplot [color=L2A1,  line width=2.0pt]
table[row sep=crcr]{%
	0.535425573221759	0.736950078740343\\
	0.876301717517881	-0.284727687864284\\
	1.11022302462516e-16	-0.882434024455501\\
	-0.902633070878332	-0.293283263199563\\
	-0.504011594351527	0.693712446173293\\
	0.535425573221758	0.736950078740343\\
};
\addlegendentry{L2A ($\beta=0.2$)}

\addplot [color=L2A2,  line width=2.0pt]
  table[row sep=crcr]{%
0.575174613431735	0.791659939041436\\
0.759687918342957	-0.246837567659802\\
1.11022302462516e-16	-0.906638439564553\\
-0.913617420848294	-0.296852294855119\\
-0.531347228598896	0.731336718935984\\
0.572032416790867	0.78733507639438\\
};
\addlegendentry{L2A ($\beta=1$)}

\node[right, align=left]
at (axis cs:0,0) {0};
\node[right, align=left]
at (axis cs:0.079,0.233) {0.25};
\node[right, align=left]
at (axis cs:0.159,0.466) {0.5};
\node[right, align=left]
at (axis cs:0.238,0.699) {0.75};
\node[right, align=left]
at (axis cs:0.430,0.850) {1};
\addplot [color=black, dotted, forget plot]
  table[row sep=crcr]{%
0	0\\
0.58283685864599	0.802206114824554\\
};
\addplot [color=black, dotted, forget plot]
  table[row sep=crcr]{%
0	-0\\
0.94304984718543	-0.306415469879979\\
};
\addplot [color=black, dotted, forget plot]
  table[row sep=crcr]{%
0	-0\\
1.11022302462516e-16	-0.99158128988915\\
};
\addplot [color=black, dotted, forget plot]
  table[row sep=crcr]{%
-0	-0\\
-0.94304984718543	-0.306415469879979\\
};
\addplot [color=black, dotted, forget plot]
  table[row sep=crcr]{%
-0	0\\
-0.58283685864599	0.802206114824554\\
};
\addplot [color=black, dotted, forget plot]
  table[row sep=crcr]{%
0	0\\
0.58283685864599	0.802206114824554\\
};
\addplot [color=black, dotted, forget plot]
  table[row sep=crcr]{%
0	0\\
0	0\\
};
\addplot [color=black, dotted, forget plot]
  table[row sep=crcr]{%
0.145709214661498	0.200551528706138\\
0.235762461796358	-0.0766038674699946\\
2.77555756156289e-17	-0.247895322472288\\
-0.235762461796357	-0.0766038674699947\\
-0.145709214661498	0.200551528706138\\
0.145709214661497	0.200551528706138\\
};
\addplot [color=black, dotted, forget plot]
  table[row sep=crcr]{%
0.291418429322995	0.401103057412277\\
0.471524923592715	-0.153207734939989\\
5.55111512312578e-17	-0.495790644944575\\
-0.471524923592715	-0.153207734939989\\
-0.291418429322995	0.401103057412277\\
0.291418429322995	0.401103057412277\\
};
\addplot [color=black, dotted, forget plot]
  table[row sep=crcr]{%
0.437127643984493	0.601654586118415\\
0.707287385389073	-0.229811602409984\\
1.11022302462516e-16	-0.743685967416863\\
-0.707287385389072	-0.229811602409984\\
-0.437127643984493	0.601654586118415\\
0.437127643984492	0.601654586118415\\
};
\addplot [color=black, dotted, forget plot]
  table[row sep=crcr]{%
0.58283685864599	0.802206114824554\\
0.94304984718543	-0.306415469879979\\
1.11022302462516e-16	-0.99158128988915\\
-0.94304984718543	-0.306415469879979\\
-0.58283685864599	0.802206114824554\\
0.58283685864599	0.802206114824554\\
};
\node[align=center,rotate=-45]
at (axis cs:0.75,0.8) {\large \textbf{Average bitrate}};
\node[align=center,,rotate=45]
at (axis cs:1,-0.4) {\large \textbf{Stability} };
\node[align=center]
at (axis cs:0,-1) {\large \textbf{Smoothness }};
\node[align=center,rotate=-45]
at (axis cs:-1,-0.4) {\large \textbf{Consistency}};
\node[align=center,rotate=45]
at (axis cs:-0.75,0.8) {\large \textbf{Continuity}};
\legend{}
\end{axis}
\end{tikzpicture}%
				\label{fig:pedestrian_radar}      
			}
		} \\ 
		\subfigure [Car user] {
			\resizebox{0.98\linewidth}{!}{%
%
\definecolor{mycolor5}{RGB}{171,217,233}
\definecolor{L2A1}{RGB}{44,123,182}

\definecolor{L2A2}{RGB}{230,97,1}

\definecolor{RB}{RGB}{253,184,99}
\definecolor{BB}{rgb}{0.24, 0.71, 0.54}
\begin{tikzpicture}

\begin{axis}[%
width=4.521in,
height=3.566in,
at={(0.758in,0.481in)},
scale only axis,
xmin=-1.39060329693044,
xmax=1.39060329693044,
ymin=-1.0967822777403,
ymax=1.0967822777403,
axis line style={draw=none},
ticks=none,
title style={font=\bfseries},
legend style={at={(0.5,1.05)}, anchor=north, legend columns=4, legend cell align=left, align=left, draw=white!15!black}
]

\addplot [color=BB, line width=2.0pt]
  table[row sep=crcr]{%
0.482430859955639	0.664009113520303\\
0.737071979242431	-0.239489203597238\\
1.11022302462516e-16	-0.745623814535469\\
-0.934792978764507	-0.303732650700817\\
-0.535123588943457	0.736534433039422\\
0.482430859955639	0.664009113520303\\
};
\addlegendentry{BB}
\addplot [color=RB,  line width=2.0pt]
  table[row sep=crcr]{%
0.311724849918447	0.42905244758934\\
0.883473907457812	-0.287058073640891\\
1.11022302462516e-16	-0.8417478568189\\
-0.913678471939966	-0.29687213155728\\
-0.42854340572437	0.589839395776166\\
0.311724849918447	0.429052447589341\\
};
\addlegendentry{RB}

\addplot [color=L2A1, line width=2.0pt]
  table[row sep=crcr]{%
0.54537593842817	0.750645581512533\\
0.884299083171765	-0.287326189483207\\
1.11022302462516e-16	-0.761647214522759\\
-0.920829530129552	-0.299195651211984\\
-0.483202877477665	0.665071704479905\\
0.54537593842817	0.750645581512533\\
};
\addlegendentry{L2A ($\beta=0.2$)}

\addplot [color=L2A2, line width=2.0pt]
  table[row sep=crcr]{%
0.575174613431735	0.791659939041436\\
0.784913995245303	-0.255034016904061\\
1.11022302462516e-16	-0.797290186890968\\
-0.920725881081009	-0.299161973594617\\
-0.502711223986173	0.691922639912502\\
0.575174613431735	0.791659939041436\\
};
\addlegendentry{L2A ($\beta=1$)}

\node[right, align=left]
at (axis cs:0,0) {0};
\node[right, align=left]
at (axis cs:0.079,0.233) {0.25};
\node[right, align=left]
at (axis cs:0.159,0.466) {0.5};
\node[right, align=left]
at (axis cs:0.238,0.699) {0.75};
\node[right, align=left]
at (axis cs:0.430,0.850) {1};
\addplot [color=black, dotted, forget plot]
  table[row sep=crcr]{%
0	0\\
0.58283685864599	0.802206114824554\\
};
\addplot [color=black, dotted, forget plot]
  table[row sep=crcr]{%
0	-0\\
0.94304984718543	-0.306415469879979\\
};
\addplot [color=black, dotted, forget plot]
  table[row sep=crcr]{%
0	-0\\
1.11022302462516e-16	-0.99158128988915\\
};
\addplot [color=black, dotted, forget plot]
  table[row sep=crcr]{%
-0	-0\\
-0.94304984718543	-0.306415469879979\\
};
\addplot [color=black, dotted, forget plot]
  table[row sep=crcr]{%
-0	0\\
-0.58283685864599	0.802206114824554\\
};
\addplot [color=black, dotted, forget plot]
  table[row sep=crcr]{%
0	0\\
0.58283685864599	0.802206114824554\\
};
\addplot [color=black, dotted, forget plot]
  table[row sep=crcr]{%
0	0\\
0	0\\
};
\addplot [color=black, dotted, forget plot]
  table[row sep=crcr]{%
0.145709214661498	0.200551528706138\\
0.235762461796358	-0.0766038674699946\\
2.77555756156289e-17	-0.247895322472288\\
-0.235762461796357	-0.0766038674699947\\
-0.145709214661498	0.200551528706138\\
0.145709214661497	0.200551528706138\\
};
\addplot [color=black, dotted, forget plot]
  table[row sep=crcr]{%
0.291418429322995	0.401103057412277\\
0.471524923592715	-0.153207734939989\\
5.55111512312578e-17	-0.495790644944575\\
-0.471524923592715	-0.153207734939989\\
-0.291418429322995	0.401103057412277\\
0.291418429322995	0.401103057412277\\
};
\addplot [color=black, dotted, forget plot]
  table[row sep=crcr]{%
0.437127643984493	0.601654586118415\\
0.707287385389073	-0.229811602409984\\
1.11022302462516e-16	-0.743685967416863\\
-0.707287385389072	-0.229811602409984\\
-0.437127643984493	0.601654586118415\\
0.437127643984492	0.601654586118415\\
};
\addplot [color=black, dotted, forget plot]
  table[row sep=crcr]{%
0.58283685864599	0.802206114824554\\
0.94304984718543	-0.306415469879979\\
1.11022302462516e-16	-0.99158128988915\\
-0.94304984718543	-0.306415469879979\\
-0.58283685864599	0.802206114824554\\
0.58283685864599	0.802206114824554\\
};
\node[align=center,rotate=-45]
at (axis cs:0.75,0.8) {\large \textbf{Average bitrate}};
\node[align=center,,rotate=45]
at (axis cs:1,-0.4) {\large \textbf{Stability}};
\node[align=center]
at (axis cs:0,-1) {\large \textbf{Smoothness}};
\node[align=center,rotate=-45]
at (axis cs:-1,-0.4) {\large \textbf{Consistency}};
\node[align=center,rotate=45]
at (axis cs:-0.75,0.8) {\large \textbf{Continuity}};
\legend{}
\end{axis}
\end{tikzpicture}%
				\label{fig:car_radar}
			}
		} &
		\subfigure [Markovian channel] {
			\resizebox{0.98\linewidth}{!}{%
%
\definecolor{mycolor5}{RGB}{171,217,233}
\definecolor{L2A1}{RGB}{44,123,182}

\definecolor{L2A2}{RGB}{230,97,1}

\definecolor{RB}{RGB}{253,184,99}
\definecolor{BB}{rgb}{0.24, 0.71, 0.54}
\begin{tikzpicture}
\begin{axis}[%
width=4.521in,
height=3.566in,
at={(0.758in,0.481in)},
scale only axis,
xmin=-1.39060329693044,
xmax=1.39060329693044,
ymin=-1.0967822777403,
ymax=1.0967822777403,
axis line style={draw=none},
ticks=none,
title style={font=\bfseries},
legend style={at={(0.5,1.05)}, anchor=north, legend columns=4, legend cell align=left, align=left, draw=white!15!black}
]
\addplot [color=BB, line width=2.0pt]
  table[row sep=crcr]{%
0.445633273085383	0.613361580235114\\
0.725041666387616	-0.235580317998864\\
1.11022302462516e-16	-0.826179743453282\\
-0.932827933066163	-0.303094168649428\\
-0.512350617597826	0.705190127003888\\
0.445633273085383	0.613361580235114\\
};
\addlegendentry{BB}
\addplot [color=RB, line width=2.0pt]
  table[row sep=crcr]{%
0.279219469928115	0.384312630252602\\
0.904218173672938	-0.293798294318084\\
1.11022302462516e-16	-0.868419688111593\\
-0.918245066482971	-0.298355908069012\\
-0.518535638668285	0.713703078183\\
0.279219469928115	0.384312630252602\\
};
\addlegendentry{RB}

\addplot [color=L2A1, line width=2.0pt]
  table[row sep=crcr]{%
0.552794673566838	0.76085659443016\\
0.906520397293352	-0.294546332117489\\
1.11022302462516e-16	-0.791211679530208\\
-0.92949923525913	-0.302012609169115\\
-0.530515116068793	0.730191414293753\\
0.552794673566838	0.76085659443016\\
};
\addlegendentry{L2A ($\beta=0.2$)}

\addplot [color=L2A2, line width=2.0pt]
  table[row sep=crcr]{%
0.587785252292473	0.809016994374947\\
0.847218292312344	-0.275277910181088\\
1.11022302462516e-16	-0.812378671242632\\
-0.928389669323452	-0.301652089342345\\
-0.542187782082172	0.746257460758265\\
0.587785252292473	0.809016994374948\\
};
\addlegendentry{L2A ($\beta=1$)}

\node[right, align=left]
at (axis cs:0,0) {0};
\node[right, align=left]
at (axis cs:0.079,0.233) {0.25};
\node[right, align=left]
at (axis cs:0.159,0.466) {0.5};
\node[right, align=left]
at (axis cs:0.238,0.699) {0.75};
\node[right, align=left]
at (axis cs:0.430,0.850) {1};
\addplot [color=black, dotted, forget plot]
  table[row sep=crcr]{%
0	0\\
0.58283685864599	0.802206114824554\\
};
\addplot [color=black, dotted, forget plot]
  table[row sep=crcr]{%
0	-0\\
0.94304984718543	-0.306415469879979\\
};
\addplot [color=black, dotted, forget plot]
  table[row sep=crcr]{%
0	-0\\
1.11022302462516e-16	-0.99158128988915\\
};
\addplot [color=black, dotted, forget plot]
  table[row sep=crcr]{%
-0	-0\\
-0.94304984718543	-0.306415469879979\\
};
\addplot [color=black, dotted, forget plot]
  table[row sep=crcr]{%
-0	0\\
-0.58283685864599	0.802206114824554\\
};
\addplot [color=black, dotted, forget plot]
  table[row sep=crcr]{%
0	0\\
0.58283685864599	0.802206114824554\\
};
\addplot [color=black, dotted, forget plot]
  table[row sep=crcr]{%
0	0\\
0	0\\
};
\addplot [color=black, dotted, forget plot]
  table[row sep=crcr]{%
0.145709214661498	0.200551528706138\\
0.235762461796358	-0.0766038674699946\\
2.77555756156289e-17	-0.247895322472288\\
-0.235762461796357	-0.0766038674699947\\
-0.145709214661498	0.200551528706138\\
0.145709214661497	0.200551528706138\\
};
\addplot [color=black, dotted, forget plot]
  table[row sep=crcr]{%
0.291418429322995	0.401103057412277\\
0.471524923592715	-0.153207734939989\\
5.55111512312578e-17	-0.495790644944575\\
-0.471524923592715	-0.153207734939989\\
-0.291418429322995	0.401103057412277\\
0.291418429322995	0.401103057412277\\
};
\addplot [color=black, dotted, forget plot]
  table[row sep=crcr]{%
0.437127643984493	0.601654586118415\\
0.707287385389073	-0.229811602409984\\
1.11022302462516e-16	-0.743685967416863\\
-0.707287385389072	-0.229811602409984\\
-0.437127643984493	0.601654586118415\\
0.437127643984492	0.601654586118415\\
};
\addplot [color=black, dotted, forget plot]
  table[row sep=crcr]{%
0.58283685864599	0.802206114824554\\
0.94304984718543	-0.306415469879979\\
1.11022302462516e-16	-0.99158128988915\\
-0.94304984718543	-0.306415469879979\\
-0.58283685864599	0.802206114824554\\
0.58283685864599	0.802206114824554\\
};
\node[align=center,rotate=-45]
at (axis cs:0.75,0.8) {\large \textbf{Average bitrate}};
\node[align=center,,rotate=45]
at (axis cs:1,-0.4) {\large \textbf{Stability}};
\node[align=center]
at (axis cs:0,-1) {\large \textbf{Smoothness}};
\node[align=center,rotate=-45]
at (axis cs:-1,-0.4) {\large \textbf{Consistency}};
\node[align=center,rotate=45]
at (axis cs:-0.75,0.8) {\large \textbf{Continuity}};
\legend{}
\end{axis}
\end{tikzpicture}%
				\label{fig:markovian_radar}
			}
		}  
	\end{tabular}
	\caption{Performance evaluation results - \emph{L2A} improves average bitrate}
	\label{fig:radar}
\end{figure*}
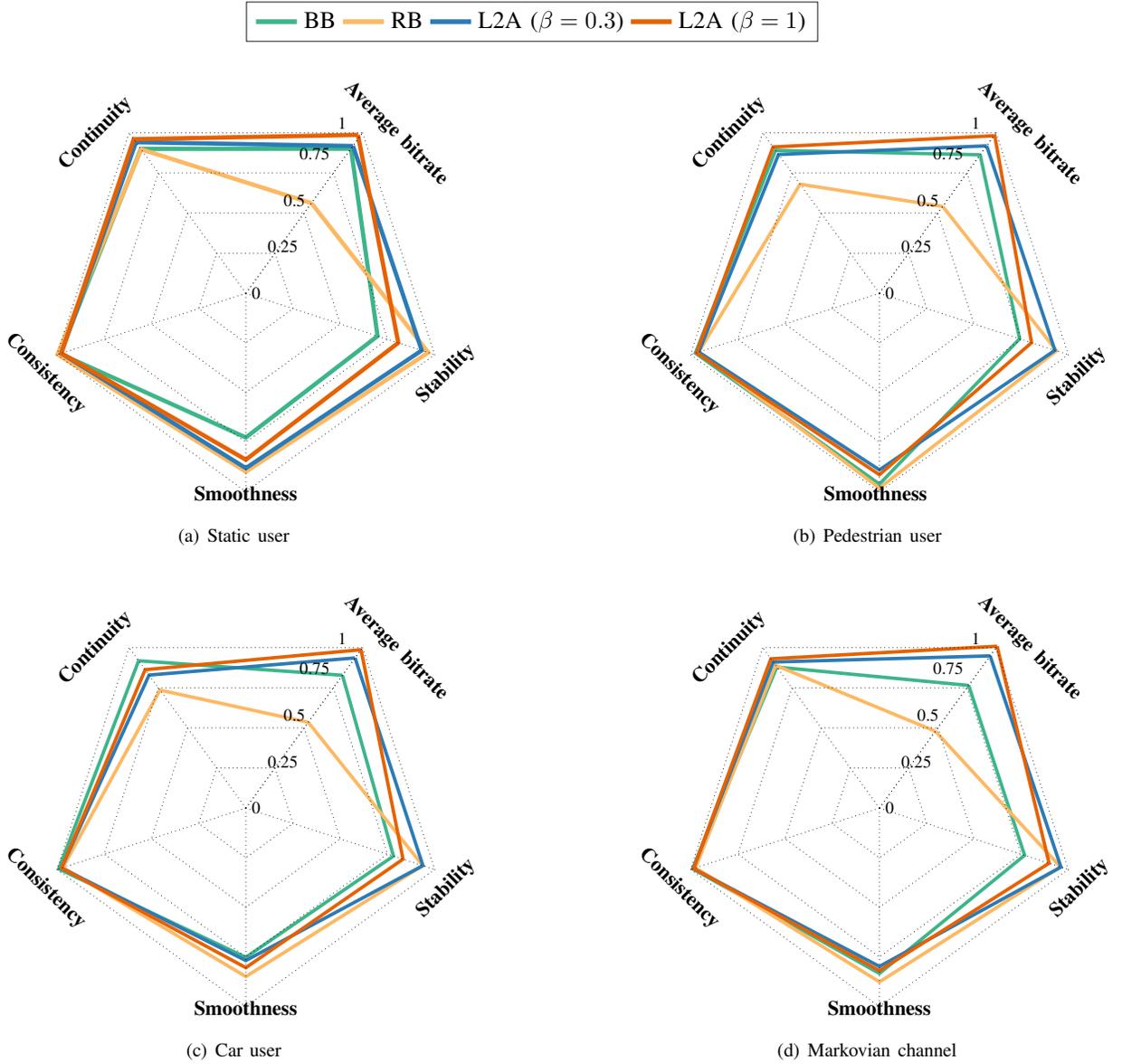

\begin{figure}[!t]
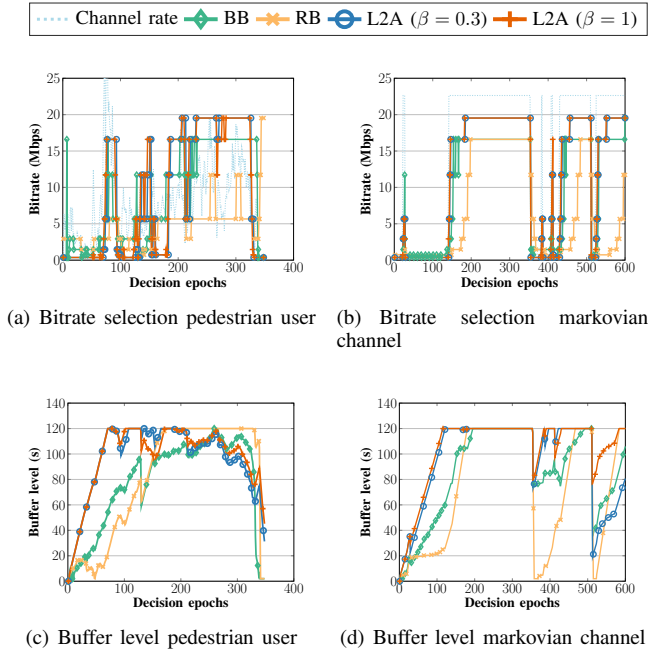
 
	\centering
	\resizebox{0.98\linewidth}{!}{%
%
%
\definecolor{mycolor5}{RGB}{171,217,233}
\definecolor{L2A1}{RGB}{44,123,182}

\definecolor{L2A2}{RGB}{230,97,1}

\definecolor{RB}{RGB}{253,184,99}
\definecolor{BB}{rgb}{0.24, 0.71, 0.54}
\begin{tikzpicture}

\begin{axis}[%
hide axis,
xmin=10,
xmax=50,
ymin=0,
ymax=0.4,
legend style={at={(0.5,1.05)}, anchor=north, legend columns=5, legend cell align=left, align=left, draw=white!15!black}
]

\addlegendimage{color=mycolor5, dotted, line width=1.5pt}
\addlegendentry{Channel rate};
\addlegendimage{color=BB, line width=2.0pt,mark size=4.0pt, mark=diamond, mark options={solid, BB}}
\addlegendentry{BB};
\addlegendimage{color=RB,  line width=2.0pt,mark size=4.0pt, mark=x, mark options={solid, RB}}
\addlegendentry{RB};
\addlegendimage{color=L2A1,  line width=2.0pt,mark size=4.0pt, mark=o, mark options={solid, L2A1}}
\addlegendentry{L2A ($\beta=0.3$)};
\addlegendimage{color=L2A2, line width=2.0pt, mark size=4.0pt, mark=+, mark options={solid, L2A2}}
\addlegendentry{L2A ($\beta=1$)};

\end{axis}
\end{tikzpicture}%
}
	
	\begin{tabular}{C{.45\linewidth}C{.45\linewidth}}
		\subfigure [Bitrate selection pedestrian user] {
			\resizebox{0.98\linewidth}{!}{%
				\input{pedestrian_Sample_path_quality_for_Bmax_120s.tex}
				\label{fig:pedestrian_quality_sample}      
			}
		} &
		\subfigure [Bitrate selection markovian channel] {
			\resizebox{0.98\linewidth}{!}{%
%
\definecolor{mycolor5}{RGB}{171,217,233}
\definecolor{L2A1}{RGB}{44,123,182}

\definecolor{L2A2}{RGB}{230,97,1}

\definecolor{RB}{RGB}{253,184,99}
\definecolor{BB}{rgb}{0.24, 0.71, 0.54}
\begin{tikzpicture}

\begin{axis}[%
width=4.521in,
height=3.566in,
at={(0.758in,0.481in)},
scale only axis,
xmin=0,
xmax=600,
xtick={0,100,200,300,400,500,600},
xlabel={\textbf{Decision epochs}},
xlabel style={font=\huge},
tick label style={font=\huge},
ymin=0,
ymax=25,
ylabel style={font=\huge},
ylabel={\textbf{Bitrate (Mbps)}},
axis background/.style={fill=white},
title style={font=\bfseries},
ymajorgrids,
legend style={at={(0.0,0.895)},anchor=west,legend cell align=left, align=left, draw=white!15!black}
]

\addplot [color=mycolor5, dotted, line width=1.5pt]
table[row sep=crcr]{%
1	0.100000000000023\\
2	0.832421875000023\\
21	0.832421875000023\\
22	13.3649739583333\\
23	22.63125\\
27	22.63125\\
28	7.09869791666665\\
29	0.832421875000023\\
140	0.832421875000023\\
141	13.3649739583333\\
142	22.63125\\
353	22.63125\\
354	7.09869791666665\\
355	0.832421875000023\\
381	0.832421875000023\\
382	7.09869791666665\\
383	22.63125\\
384	22.63125\\
385	13.3649739583333\\
386	0.832421875000023\\
406	0.832421875000023\\
407	7.09869791666665\\
408	22.63125\\
411	22.63125\\
412	7.09869791666665\\
413	0.832421875000023\\
428	0.832421875000023\\
429	13.3649739583333\\
430	22.63125\\
510	22.63125\\
511	7.09869791666665\\
512	0.832421875000023\\
523	0.832421875000023\\
524	13.3649739583333\\
525	22.63125\\
600	22.63125\\
};
\addlegendentry{Channel rate}

\addplot [color=BB, line width=2pt, mark size=4.0pt, mark=diamond, mark options={solid, BB},mark repeat=1]
table[row sep=crcr]{%
1	0.3662109375\\
22  0.3662109375\\
23	1.46484375\\
24	1.46484375\\
25	2.9296875\\
26	5.6640625\\
27	11.71875\\
28	2.9296875\\
29	0.732421875\\
30	0.3662109375\\
40	0.3662109375\\
41	0.732421875\\
50	0.732421875\\
51	0.3662109375\\
60	0.3662109375\\
61	0.732421875\\
70  0.732421875\\
71	0.3662109375\\
80	0.3662109375\\
81	0.732421875\\
90	0.732421875\\
91	0.3662109375\\
100	0.3662109375\\
101	0.732421875\\
110	0.732421875\\
113	0.3662109375\\
120	0.3662109375\\
121	0.732421875\\
136	0.732421875\\
137	1.46484375\\
144	1.46484375\\
145	2.9296875\\
148	2.9296875\\
149	5.6640625\\
152	5.6640625\\
153	11.71875\\
154	16.6015625\\
155	11.71875\\
159	11.71875\\
160	16.6015625\\
166	16.6015625\\
167	11.71875\\
168	16.6015625\\
354	16.6015625\\
355	1.46484375\\
356	1.46484375\\
357	0.732421875\\
358	1.46484375\\
359	0.732421875\\
362	0.732421875\\
363	1.46484375\\
406	1.46484375\\
407	2.9296875\\
410	2.9296875\\
411	5.6640625\\
412	5.6640625\\
413	1.46484375\\
426	1.46484375\\
427	0.732421875\\
428	0.732421875\\
429	1.46484375\\
434	1.46484375\\
435	2.9296875\\
437	2.9296875\\
438	5.6640625\\
441	5.6640625\\
442	16.6015625\\
443	11.71875\\
444	16.6015625\\
445	16.6015625\\
446	11.71875\\
447	16.6015625\\
512	16.6015625\\
513	0.3662109375\\
518	0.3662109375\\
519	0.732421875\\
525	0.732421875\\
526	1.46484375\\
528	1.46484375\\
529	2.9296875\\
530	5.6640625\\
531	5.6640625\\
532	16.6015625\\
599	16.6015625\\
600	19.53125\\
};
\addlegendentry{BB}

\addplot [color=RB, line width=2pt,mark size=4.0pt, mark=x, mark options={solid, RB},mark repeat=1]
table[row sep=crcr]{%
1	0.3662109375\\
2	0.732421875\\
27	0.732421875\\
28	1.46484375\\
29	0.732421875\\
149	0.732421875\\
150	1.46484375\\
173	1.46484375\\
174	2.9296875\\
180	2.9296875\\
181	5.6640625\\
190	5.6640625\\
191	11.71875\\
197	11.71875\\
198	16.6015625\\
353	16.6015625\\
354	11.71875\\
357	11.71875\\
358	5.6640625\\
363	5.6640625\\
364	2.9296875\\
367	2.9296875\\
368	1.46484375\\
391	1.46484375\\
392	0.732421875\\
433	0.732421875\\
434	1.46484375\\
458	1.46484375\\
459	2.9296875\\
466	2.9296875\\
467	5.6640625\\
476	5.6640625\\
477	11.71875\\
483	11.71875\\
484	16.6015625\\
510	16.6015625\\
511	11.71875\\
514	11.71875\\
515	5.6640625\\
520	5.6640625\\
521	2.9296875\\
524	2.9296875\\
525	1.46484375\\
528	1.46484375\\
529	0.732421875\\
560	0.732421875\\
561	1.46484375\\
578	1.46484375\\
579	2.9296875\\
585	2.9296875\\
586	5.6640625\\
595	5.6640625\\
596	11.71875\\
600	11.71875\\
};
\addlegendentry{RB}

\addplot [color=L2A1, line width=2pt,mark size=4.0pt, mark=o, mark options={solid, L2A1},mark repeat=1]
table[row sep=crcr]{%
1	0.3662109375\\
22	0.3662109375\\
23	2.9296875\\
24	2.9296875\\
25	5.6640625\\
28	5.6640625\\
29	0.3662109375\\
142	0.3662109375\\
143	5.6640625\\
144	5.6640625\\
145	11.71875\\
146	11.71875\\
147	16.6015625\\
184	16.6015625\\
185	19.53125\\
354	19.53125\\
355	0.3662109375\\
382	0.3662109375\\
383	2.9296875\\
384	2.9296875\\
385	5.6640625\\
386	5.6640625\\
387	0.3662109375\\
408	0.3662109375\\
409	5.6640625\\
410	5.6640625\\
411	11.71875\\
412	11.71875\\
413	0.732421875\\
414	0.732421875\\
415	0.3662109375\\
430	0.3662109375\\
431	2.9296875\\
432	2.9296875\\
433	5.6640625\\
434	5.6640625\\
435	11.71875\\
436	11.71875\\
437	16.6015625\\
456	16.6015625\\
457	19.53125\\
512	19.53125\\
513	0.3662109375\\
524	0.3662109375\\
525	2.9296875\\
526	2.9296875\\
527	5.6640625\\
528	5.6640625\\
529	11.71875\\
530	11.71875\\
531	16.6015625\\
552	16.6015625\\
553	19.53125\\
600	19.53125\\
};
\addlegendentry{L2A $\beta=0.2$}

\addplot [color=L2A2, line width=2pt,mark size=4.0pt, mark=+, mark options={solid, L2A2},mark repeat=1]
table[row sep=crcr]{%
1	0.3662109375\\
21	0.3662109375\\
22	0.732421875\\
23	2.9296875\\
24	2.9296875\\
25	5.6640625\\
27	5.6640625\\
28	2.9296875\\
29	0.3662109375\\
140	0.3662109375\\
141	2.9296875\\
142	5.6640625\\
143	5.6640625\\
144	11.71875\\
145	11.71875\\
146	16.6015625\\
182	16.6015625\\
183	19.53125\\
353	19.53125\\
354	16.6015625\\
355	0.3662109375\\
381	0.3662109375\\
382	0.732421875\\
383	2.9296875\\
384	5.6640625\\
385	5.6640625\\
386	0.3662109375\\
406	0.3662109375\\
407	1.46484375\\
408	2.9296875\\
409	5.6640625\\
410	11.71875\\
411	11.71875\\
412	16.6015625\\
413	0.732421875\\
414	0.3662109375\\
428	0.3662109375\\
429	1.46484375\\
430	2.9296875\\
431	5.6640625\\
432	5.6640625\\
433	11.71875\\
434	11.71875\\
435	16.6015625\\
455	16.6015625\\
456	19.53125\\
510	19.53125\\
511	16.6015625\\
512	0.732421875\\
513	0.3662109375\\
523	0.3662109375\\
524	1.46484375\\
525	2.9296875\\
526	5.6640625\\
527	5.6640625\\
528	11.71875\\
529	11.71875\\
530	16.6015625\\
550	16.6015625\\
551	19.53125\\
600	19.53125\\
};
\addlegendentry{L2A $\beta=1$}

\legend{}
\end{axis}
\end{tikzpicture}%
				\label{fig:markovian_quality_sample}
			}
		} \\
		\subfigure [Buffer level pedestrian user] {
			\resizebox{0.98\linewidth}{!}{%
				\input{pedestrian_Sample_path_buffer_for_Bmax_120s.tex}
				\label{fig:pedestrian_buffer_sample}
			}
		} & 
		\subfigure [Buffer level markovian channel] {
			\resizebox{0.98\linewidth}{!}{%
				\input{markovian_Sample_path_buffer_for_Bmax_120s.tex}
				\label{fig:markovian_buffer_sample}      
			}
		} \\
	\end{tabular}
	\caption{Sample paths for bitrate selection and buffer level}
	\label{fig:sample}
\end{figure}

\begin{figure}[!t]
	\centering
	\resizebox{0.98\linewidth}{!}{%
%
%
\definecolor{mycolor5}{RGB}{171,217,233}
\definecolor{L2A1}{RGB}{44,123,182}

\definecolor{L2A2}{RGB}{230,97,1}

\definecolor{RB}{RGB}{253,184,99}
\definecolor{BB}{rgb}{0.24, 0.71, 0.54}
\begin{tikzpicture}

\begin{axis}[%
hide axis,
xmin=10,
xmax=50,
ymin=0,
ymax=0.4,
legend style={at={(0.5,1.05)}, anchor=north, legend columns=5, legend cell align=left, align=left, draw=white!15!black}
]

\addlegendimage{color=BB, line width=2.0pt,mark size=4.0pt, mark=diamond, mark options={solid, BB}}
\addlegendentry{BB};
\addlegendimage{color=RB,  line width=2.0pt,mark size=4.0pt, mark=x, mark options={solid, RB}}
\addlegendentry{RB};
\addlegendimage{color=L2A1,  line width=2.0pt,mark size=4.0pt, mark=o, mark options={solid, L2A1}}
\addlegendentry{L2A ($\beta=0.3$)};
\addlegendimage{color=L2A2, line width=2.0pt, mark size=4.0pt, mark=+, mark options={solid, L2A2}}
\addlegendentry{L2A ($\beta=1$)};

\end{axis}
\end{tikzpicture}%
}

	\begin{tabular}{C{.45\linewidth}C{.45\linewidth}}
		\subfigure [$T^{0.9}$-slot regret $R_T/T$ ] {
			\resizebox{0.98\linewidth}{!}{%
%
%
%
\definecolor{mycolor5}{RGB}{171,217,233}
\definecolor{L2A1}{RGB}{44,123,182}

\definecolor{L2A2}{RGB}{230,97,1}

\definecolor{RB}{RGB}{253,184,99}
\definecolor{BB}{rgb}{0.24, 0.71, 0.54}
\begin{tikzpicture}

\begin{axis}[%
width=4.521in,
height=3.565in,
at={(0.758in,0.482in)},
scale only axis,
xmin=0,
xmax=900,
xlabel style={font=\huge},
xlabel={\textbf{Horizon (minutes)}},
tick label style={font=\huge},
ymin=-7,
ymax=2,
ylabel style={font=\huge},
ylabel={\textbf{$T^{0.9}$-slot regret} $R_T/T$},
axis background/.style={fill=white},
ymajorgrids,
legend style={legend cell align=left, align=left, draw=white!15!black}
]
\addplot [color=BB,  line width=2pt, mark size=4.0pt, mark=diamond, mark options={solid, BB},smooth]
  table[row sep=crcr]{%
23.2	-4.31848284841954\\
55.2333333333333	-4.8345737307634\\
72.4333333333334	-4.09346192979172\\
99.4666666666667	-4.11585585361513\\
130.4	-1.17706735685067\\
143.066666666667	-0.90114739267824\\
154.133333333333	-0.838388621188415\\
175.466666666667	-0.755857334310917\\
200.233333333333	-0.961382456144861\\
245.6	-1.60422392440285\\
291.666666666667	-1.12235770089285\\
333.366666666667	-1.38886257858587\\
342.066666666667	-1.46988738793607\\
364.733333333333	-1.66412113658953\\
372.666666666667	-1.83534571011296\\
399.5	-1.87315646837192\\
426.133333333333	-1.87941301749254\\
445.966666666667	-2.15217470008974\\
469.966666666667	-2.21754285543307\\
504.5	-2.30726843822265\\
530.933333333333	-2.23134537547867\\
549.133333333333	-2.10221468772761\\
555.733333333333	-2.13419352291646\\
572.033333333333	-2.18180428642711\\
582.4	-2.16309543931004\\
604.133333333333	-2.24326088896282\\
631.066666666667	-2.46863543496136\\
652.833333333333	-2.34962632834436\\
668.133333333333	-2.4175244637579\\
707.366666666667	-2.37273256738604\\
729.433333333333	-2.42817559756713\\
752.333333333333	-2.44775152650084\\
766.933333333333	-2.39113289059787\\
784.333333333333	-2.37385410579577\\
809.666666666667	-2.34181717817262\\
825.966666666667	-2.31503029751752\\
858.533333333333	-2.45656787951111\\
871.166666666667	-2.46437237690122\\
889.4	-2.50794513857659\\
939.5	-2.57226857010596\\
973.9	-2.57352729815898\\
};
\addlegendentry{BB}
\addplot [color=RB,  line width=2pt, mark size=4.0pt, mark=x, mark options={solid, RB},smooth]
  table[row sep=crcr]{%
23.2	0.36375550017965\\
55.2333333333333	-1.77572891520822\\
72.4333333333334	-1.28666491098136\\
99.4666666666667	-1.51512002817105\\
130.4	1.26700372052335\\
143.066666666667	1.45062875881001\\
154.133333333333	1.46983321562766\\
175.466666666667	1.4958483107547\\
200.233333333333	1.38385506622478\\
245.6	0.877107005165954\\
291.666666666667	1.03530412946429\\
333.366666666667	0.858725162639985\\
342.066666666667	0.832431029678673\\
364.733333333333	0.666808245321931\\
372.666666666667	0.62609404698685\\
399.5	0.595244788407399\\
426.133333333333	0.652065921635426\\
445.966666666667	0.359718282266044\\
469.966666666667	0.309448545220448\\
504.5	0.261465172819612\\
530.933333333333	0.343305126977612\\
549.133333333333	0.442491174540237\\
555.733333333333	0.386439815821404\\
572.033333333333	0.391700521895586\\
582.4	0.416196047604785\\
604.133333333333	0.319715915360803\\
631.066666666667	0.129263511349791\\
652.833333333333	0.22844631573912\\
668.133333333333	0.191200452753947\\
707.366666666667	0.198625845271181\\
729.433333333333	0.152754403389622\\
752.333333333333	0.10787835414817\\
766.933333333333	0.136930123489606\\
784.333333333333	0.158534732920771\\
809.666666666667	0.115822563940924\\
825.966666666667	0.143695069098385\\
858.533333333333	-0.00278018501751376\\
871.166666666667	-0.0456790986105489\\
889.4	-0.0570750309782397\\
939.5	-0.161649072256978\\
973.9	-0.170858328145869\\
};
\addlegendentry{RB}

\addplot [color=L2A1,  line width=2pt, mark size=4.0pt, mark=o, mark options={solid, L2A1}]
table[row sep=crcr]{%
23.2	-6.59281412760413\\
55.2333333333333	-6.69365228481445\\
72.4333333333334	-5.95460444949379\\
99.4666666666667	-6.19626288119969\\
130.4	-3.00170648804954\\
143.066666666667	-2.49327760040478\\
154.133333333333	-2.21162683823525\\
175.466666666667	-2.11390811259264\\
200.233333333333	-2.21700171414602\\
245.6	-2.75016395328612\\
291.666666666667	-2.219609375\\
333.366666666667	-2.44751403765872\\
342.066666666667	-2.39227122438979\\
364.733333333333	-2.64090006540164\\
372.666666666667	-2.73784494144115\\
399.5	-2.69731492164681\\
426.133333333333	-2.68053572825409\\
445.966666666667	-2.94880232956029\\
469.966666666667	-3.00408467577483\\
504.5	-3.08012071048063\\
530.933333333333	-2.9294913046208\\
549.133333333333	-2.72267776828255\\
555.733333333333	-2.78103877852834\\
572.033333333333	-2.74801677176811\\
582.4	-2.73179833268944\\
604.133333333333	-2.80213735930261\\
631.066666666667	-3.05169978203242\\
652.833333333333	-2.83820176873246\\
668.133333333333	-2.91723564603183\\
707.366666666667	-2.82409992975704\\
729.433333333333	-2.90578218001701\\
752.333333333333	-2.89564627063021\\
766.933333333333	-2.79107657526731\\
784.333333333333	-2.76196091923873\\
809.666666666667	-2.71585388309234\\
825.966666666667	-2.69980179474055\\
858.533333333333	-2.86131750854167\\
871.166666666667	-2.89801787892191\\
889.4	-2.93376474542765\\
939.5	-3.01101400345931\\
973.9	-3.04792653356776\\
};
\addlegendentry{L2A$\beta=0.2$}
\addplot [color=L2A2, line width=2pt, mark size=4.0pt, mark=+, mark options={solid, L2A2},smooth]
  table[row sep=crcr]{%
23.2	-6.7081495263111\\
55.2333333333333	-6.77277329605465\\
72.4333333333334	-6.06205778301887\\
99.4666666666667	-6.25252813180714\\
130.4	-3.1447209707311\\
143.066666666667	-2.69407928267708\\
154.133333333333	-2.38731938662409\\
175.466666666667	-2.28166203368403\\
200.233333333333	-2.34761511960005\\
245.6	-2.90175447246543\\
291.666666666667	-2.23897321428569\\
333.366666666667	-2.56287847777719\\
342.066666666667	-2.54454626671827\\
364.733333333333	-2.75439578073588\\
372.666666666667	-2.87774974843467\\
399.5	-2.85593808667079\\
426.133333333333	-2.84164676976593\\
445.966666666667	-3.13396526564486\\
469.966666666667	-3.16926012039858\\
504.5	-3.23210732315408\\
530.933333333333	-3.07953018809246\\
549.133333333333	-2.86268887479889\\
555.733333333333	-2.9161215286108\\
572.033333333333	-2.89829534901082\\
582.4	-2.8967679201902\\
604.133333333333	-2.97422255408924\\
631.066666666667	-3.20840654629069\\
652.833333333333	-3.02032681061723\\
668.133333333333	-3.06450930024255\\
707.366666666667	-2.98154365442838\\
729.433333333333	-3.04745693003701\\
752.333333333333	-3.03810562451542\\
766.933333333333	-2.93273395225299\\
784.333333333333	-2.88611675088987\\
809.666666666667	-2.84895625611102\\
825.966666666667	-2.81584322599281\\
858.533333333333	-2.98391390446648\\
871.166666666667	-3.00604788687349\\
889.4	-3.06729727818617\\
939.5	-3.13541268211145\\
973.9	-3.16200756408944\\
};
\addlegendentry{L2A $(\beta=1)$}
\legend{}
\end{axis}

\end{tikzpicture}%
				\label{fig:regret}
			}
		}&
		\subfigure [Constraint residual $V^1_T/T$] {
			\resizebox{0.98\linewidth}{!}{%
				\definecolor{mycolor5}{RGB}{171,217,233}
\definecolor{L2A1}{RGB}{44,123,182}

\definecolor{L2A2}{RGB}{230,97,1}

\definecolor{RB}{RGB}{253,184,99}
\definecolor{BB}{rgb}{0.24, 0.71, 0.54}
\begin{tikzpicture}

\begin{axis}[%
width=4.521in,
height=3.765in,
at={(0.758in,0.482in)},
scale only axis,
xmin=0,
xmax=900,
xlabel style={font=\huge},
xlabel={\textbf{Horizon (minutes)}},
tick label style={font=\huge},
ymin=-0.5000,
ymax=2.5,
ylabel style={font=\huge},
ylabel={\textbf{Constraint residual $V^1_T$/$T$}},
axis background/.style={fill=white},
ymajorgrids,
legend style={legend cell align=left, align=left, draw=white!15!black}
]
\addplot [color=BB, line width=2pt, mark size=4.0pt, mark=diamond, mark options={solid, BB}]
  table[row sep=crcr]{%
23.2	-0.0164422219162361\\
55.2333333333333	0.0794683885150107\\
72.4333333333334	0.0959913968220008\\
99.4666666666667	0.0341676674599967\\
130.4	0.0295756188763789\\
143.066666666667	0.0278863523016071\\
154.133333333333	0.0365524629959282\\
175.466666666667	0.0260966862157375\\
200.233333333333	0.0147329494711812\\
245.6	0.0311335640304833\\
291.666666666667	0.073594718590698\\
333.366666666667	0.0161125728385514\\
342.066666666667	0.0188910833899172\\
364.733333333333	0.00679812293060422\\
372.666666666667	0.0136842336728478\\
399.5	0.0260254663783144\\
426.133333333333	0.00687746634935138\\
445.966666666667	0.00729641471730247\\
469.966666666667	0.000846225098825926\\
504.5	0.00315654861492476\\
530.933333333333	0.00735832638508782\\
549.133333333333	0.0238018026426516\\
555.733333333333	-0.00125353920975613\\
572.033333333333	0.00980151398402995\\
582.4	0.00814330248829265\\
604.133333333333	0.0208635769802186\\
631.066666666667	0.0203713602753623\\
652.833333333333	0.0254288249786896\\
668.133333333333	0.0228459336302649\\
707.366666666667	0.0165761646288729\\
729.433333333333	0.00735623490709258\\
752.333333333333	0.00159207035744657\\
766.933333333333	0.00302003916317517\\
784.333333333333	0.00209517079474608\\
809.666666666667	-0.00642932037112587\\
825.966666666667	0.00286200287359861\\
858.533333333333	-0.0115871374790686\\
871.166666666667	-0.0214039964050698\\
889.4	-0.0153769965197625\\
939.5	-0.024961088009718\\
973.9	-0.0247525945037523\\
};
\addlegendentry{BB}

\addplot [color=RB, line width=2pt, mark size=4.0pt, mark=x, mark options={solid, RB}]
table[row sep=crcr]{%
23.2	-0.473206834719122\\
55.2333333333333	-0.170737732728071\\
72.4333333333334	-0.228338960992119\\
99.4666666666667	0.0941109553737078\\
130.4	0.0175427697919304\\
143.066666666667	0.00495702466082548\\
154.133333333333	-0.00588653010595408\\
175.466666666667	1.00647799128797\\
200.233333333333	0.823087682456617\\
245.6	0.9228178097311\\
291.666666666667	2.3106449908978\\
333.366666666667	2.3017421899134\\
342.066666666667	2.22575987612527\\
364.733333333333	2.06464625465128\\
372.666666666667	1.99923441006706\\
399.5	1.84567510113345\\
426.133333333333	1.70652829155529\\
445.966666666667	1.61924807272965\\
469.966666666667	1.64129129755781\\
504.5	1.49994561331994\\
530.933333333333	1.40818619295464\\
549.133333333333	1.38528324121341\\
555.733333333333	1.36001406484786\\
572.033333333333	1.31450367111506\\
582.4	1.27542476973633\\
604.133333333333	1.22433195065105\\
631.066666666667	1.15949490694516\\
652.833333333333	1.1247965035476\\
668.133333333333	1.08185180336295\\
707.366666666667	1.27870797808646\\
729.433333333333	1.21500525803947\\
752.333333333333	1.16029533727817\\
766.933333333333	1.13386296405531\\
784.333333333333	1.09991821050858\\
809.666666666667	1.19946431594394\\
825.966666666667	1.23522976954769\\
858.533333333333	1.17076002369743\\
871.166666666667	1.14252782669553\\
889.4	1.12040420292681\\
939.5	1.04358584509646\\
973.9	0.998730904743979\\
};
\addlegendentry{RB}

\addplot [color=L2A1, line width=2pt, mark size=4.0pt, mark=o, mark options={solid, L2A1}]
  table[row sep=crcr]{%
23.2	0.0337111829253445\\
55.2333333333333	0.158727401435499\\
72.4333333333334	0.216554494587172\\
99.4666666666667	0.242217443995628\\
130.4	0.246116866549869\\
143.066666666667	0.354699844895435\\
154.133333333333	0.27777255461092\\
175.466666666667	0.287452413508504\\
200.233333333333	0.245109458092657\\
245.6	0.210049397161583\\
291.666666666667	0.17052172749095\\
333.366666666667	0.119327613788982\\
342.066666666667	0.0818074275126719\\
364.733333333333	0.108617219400912\\
372.666666666667	0.0883967978801365\\
399.5	0.064512664445715\\
426.133333333333	0.0655532453221213\\
445.966666666667	0.0876490040361659\\
469.966666666667	0.0810504913105206\\
504.5	0.0836405203225468\\
530.933333333333	0.0674289869696167\\
549.133333333333	0.0516985019155527\\
555.733333333333	0.0440604078285105\\
572.033333333333	0.0624433436489653\\
582.4	0.0618339622393478\\
604.133333333333	0.0834503535149906\\
631.066666666667	0.0905284585759318\\
652.833333333333	0.0537537396935477\\
668.133333333333	0.0568269919737077\\
707.366666666667	0.0219255224528752\\
729.433333333333	0.0347475825905121\\
752.333333333333	0.0270179427062658\\
766.933333333333	0.00729410743849712\\
784.333333333333	0.00623467253058152\\
809.666666666667	0.00431337059183079\\
825.966666666667	0.0131339043078924\\
858.533333333333	0.0158039677703528\\
871.166666666667	0.0122291917807615\\
889.4	0.0264482553239986\\
939.5	0.0382228700117366\\
973.9	0.0525698580559038\\
};
\addlegendentry{L2A $(\beta=0.2)$}

\addplot [color=L2A2, line width=2pt, mark size=4.0pt, mark=+, mark options={solid, L2A2}]
  table[row sep=crcr]{%
23.2	-0.093606457170722\\
55.2333333333333	0.262028560916519\\
72.4333333333334	0.275876042497998\\
99.4666666666667	0.25106684167713\\
130.4	0.222072535145685\\
143.066666666667	0.241363451608322\\
154.133333333333	0.267718333865332\\
175.466666666667	0.293224363754007\\
200.233333333333	0.221740331985757\\
245.6	0.23357925835171\\
291.666666666667	0.185745816157919\\
333.366666666667	0.17570052937981\\
342.066666666667	0.0946351583800151\\
364.733333333333	0.0849085993261269\\
372.666666666667	0.0664459909944526\\
399.5	0.0636977236774783\\
426.133333333333	0.0698835320388298\\
445.966666666667	0.091025509453516\\
469.966666666667	0.0825497736988154\\
504.5	0.0925389421651062\\
530.933333333333	0.0648053109565581\\
549.133333333333	0.0535608971198371\\
555.733333333333	0.0481582872304216\\
572.033333333333	0.0526266416729868\\
582.4	0.0562544857415332\\
604.133333333333	0.0837381702529001\\
631.066666666667	0.0875468974007845\\
652.833333333333	0.0679464050203933\\
668.133333333333	0.069783038577043\\
707.366666666667	0.0590084118769028\\
729.433333333333	0.0680416928432805\\
752.333333333333	0.0527250693321548\\
766.933333333333	0.0469906270700449\\
784.333333333333	0.0385381152163973\\
809.666666666667	0.0318928135546912\\
825.966666666667	0.0370400352589968\\
858.533333333333	0.0453842361689567\\
871.166666666667	0.0319446628882361\\
889.4	0.0455522429634811\\
939.5	0.0524425755706943\\
973.9	0.0651649403723695\\
};
\addlegendentry{L2A $(\beta=1)$}
\legend{}
\end{axis}
\end{tikzpicture}%
				\label{fig:violation}
		}}
	\end{tabular}
	\caption{Convergence of regret and constraint residual}
	\label{fig:regret_const}
\end{figure}
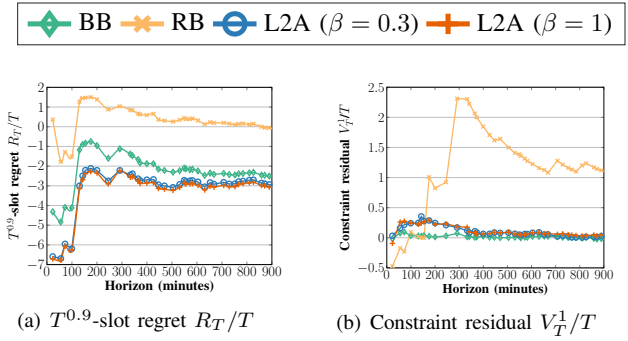

\begin{table*}[t]\caption{Live streaming ($B_{max}=20s$) results (BB / RB / L2A ($\beta=0.3$) / L2A ($\beta=1$))}
	\centering
	\begin{tabular}{c|cccc}
		& Static & Pedestrian    &   Car  &   Markovian \\ \hline
		
		\multirow{ 1}{*}{\textbf{Average bitrate} } &   0.93 / 0.59 / 0.88 / \textbf{0.98}  &  0.94 / 0.58 / 0.93 / \textbf{0.96}     & 0.93 / 0.59 / 0.96 / \textbf{0.98} & 0.91 / 0.69 / 0.97 / \textbf{1.00}      \\
		
		\multirow{ 1}{*}{\textbf{Stability} }     &  0.42 / \textbf{0.95} / 0.82 / 0.72 & 0.50 / \textbf{0.92} / 0.86 / 0.75    & 0.56 / \textbf{0.91} / 0.86 / 0.78    & 0.86 / \textbf{0.93} / 0.87 / 0.82  \\
		
		\multirow{ 1}{*}{\textbf{Smoothness} }   &  0.86 / 0.92 / 0.94 / \textbf{0.94} &  0.86 / 0.92 / 0.94 / \textbf{0.95}   & 0.87 / 0.95 / 0.94 / \textbf{0.97}    & 0.96 / 0.95 / 0.92 / \textbf{0.98}    \\
		
		\multirow{ 1}{*}{\textbf{Consistency}} &  0.87 / 0.81 / 0.92 / \textbf{0.92} & 0.85 / 0.62 / 0.83 / \textbf{0.85} & 0.88 / 0.62 / 0.80 / \textbf{0.85}   & 0.78 / 0.48 / 0.83 / \textbf{0.84}   \\  
		
		\multirow{ 1}{*}{\textbf{Continuity}} & 0.93 / 0.95 / 0.97 / \textbf{0.97}  & 0.93 / 0.93 / 0.97 / \textbf{0.97}    & 0.93 / 0.94 / 0.95 / \textbf{0.95}    & 0.92 / 0.88 / 0.94 / \textbf{0.94}    
	\end{tabular}
	\label{tab:live} 
\end{table*}

In regard to the evaluation results, each point on \fref{fig:radar} corresponds to a score for each of the five performance metrics and is the average result over all traces for the considered network scenario. In particular, \fref{fig:static_radar} shows the performance of a static user (no mobility), \fref{fig:pedestrian_radar} shows the performance of a pedestrian user (low mobility), \fref{fig:car_radar} shows a user while being mobile in a car (high mobility) and \fref{fig:markovian_radar} corresponds to the artificial \emph{markovian} scenario. 

In \fref{fig:radar} \emph{L2A}, our proposed method, registers significant improvement in average bitrate, almost up to $45\%$ against \emph{RB} and up to $20\%$ against \emph{BB}, for all studied real network scenarios and for both the cases of restricted ($\beta=0.3$) and unrestricted ($\beta=1$) switching. At the same time \emph{L2A} offers consistent (i.e without interruptions) streaming with equivalent continuity to all the other methods, i.e. all methods experience a few brief stalls during periods of very poor channel quality. In regard to smoothness, all methods obtain equivalent scores. Nonetheless, in regard to the stability metric, \emph{L2A}'s restricted switching variant ($\beta=0.3$) achieves about $15\%$ improvement in stability when compared to the case of unrestricted ($\beta=1$) switching; a result that is anticipated from our algorithmic design. Additionally, \emph{L2A}'s restricted switching variant ($\beta=0.3$) improves on stability by $25\%$ against \emph{BB}. Comparing \emph{L2A} and \emph{RB} in stability, we observe equivalent performance, yet \emph{RB} is overall more conservative, given the low average bitrate it obtained in all scenarios.

In the markovian network scenario depicted in \fref{fig:markovian_radar}, \emph{L2A} performs $50\%$ better against \emph{RB} and $25\%$ better against \emph{BB} in terms of average bitrate, while performing equally well, or even better (i.e. against \emph{BB} in stability), in all other metrics. Thus \emph{L2A} is \emph{robust} against the channel fluctuations and doesn't require any assumption on the channel rate distribution.

In \fref{fig:pedestrian_quality_sample}, a sample path for the channel rate and the bitrate selection for each method is presented for a randomly-selected pedestrian-mobility trace, while, \fref{fig:pedestrian_buffer_sample} depicts the evolution of the buffer for the same trace\footnote{We note here that in \fref{fig:sample}, while some markers have been omitted for clarity, the lines remain an accurate representation of the results.}. From these plots, we can argue that \emph{L2A} learns the volatile channel distribution, in order to re-actively provide the highest bitrate (which is the optimization objective) and to pro-actively protect the buffer from under-flowing (which is one of the optimization constraints). While this `adaptive behavior' of \emph{L2A} may come only at a marginal cost in smoothness (i.e. bitrate distance between consecutive decisions), \fref{fig:pedestrian_radar} shows that L2A achieves on average only $3\%$ less smoothness than the other methods;  a trade-off that is aligned with common \ac{HAS} optimization principles. 

Similarly, we provide a bitrate selection sample path in \fref{fig:markovian_quality_sample}, for a randomly-selected markovian trace and its corresponding buffer evolution in \fref{fig:markovian_buffer_sample}. Here, we observe: (i) that in terms of matching the bitrate to the channel rate at each decision epoch, \emph{L2A} presents a more efficient channel utilization, (ii) a slightly unstable behavior for \emph{BB}, especially at the beginning of the session and (iii) some stall events occurring for \emph{RB}. \emph{L2A} consistently manages to offer high bitrate, stable and uninterrupted streaming; even in the most demanding network scenarios.

To further investigate the \emph{robustness} property of our method, we have synthesized an additional network profile, where we have concatenated \emph{car} traces to extend the streaming session duration (horizon) in order to simulate longer, yet realistic scenarios. For these concatenated \emph{car} traces, \fref{fig:regret} presents the regret rate $R_T/T$ against the K-Slot benchmark of \sref{Algorithm}, for $K=T^{0.9}$. Here \emph{L2A}, for both studied values of $\beta~(0.3, 1)$, achieves better regret than any other method, significantly improving on the K-Benchmark in any streaming horizon, a result that is anticipated from Theorem \ref{thm:theorem1}.

Regarding the constraint residual $V^i_T~\forall i=1,2$, we examine in particular the case of underflow ($i=1$), as stalls are the most significant factors that can affect the streaming experience. Potential buffer overflows ($i=2$) can be easily tackled by simply inducing a short delay before requests, according to \eqref{eq:buffer_evolution}. In \fref{fig:violation} we present the constraint residual rate for the concatenated \emph{car} traces. \emph{L2A} manages to respect the underflow constraint on average, given that the constraint residual rate $V_T^1/T$ converges to 0.

In order to investigate the merits of \emph{L2A} beyond \ac{VoD}, we repeated the same cycle of experiments for the case of \emph{live streaming}, where now $B_{max}=20$s. In industrial live streaming applications, such small buffer values are commonly used, given the strict delay requirements. We present our results in \tref{tab:live}, which depicts that although \emph{RB} achieves higher values in stability, it is not able to compete with the other methods in terms of average bitrate. On the contrary, \emph{L2A} manages to provide up to $30\%$ higher live streaming bitrate when compared to \emph{RB} and equivalent bitrate to \emph{BB}, while its switch-restricted instance shows up to $40\%$ improvement in stability when compared to \emph{BB}. Online learning methods have -- by design -- less dependency on the instantaneous buffer length and are also more reactive to throughput fluctuations, unlike throughput-based methods that are, normally, as efficient as their throughput estimation module.

	\section{Conclusions}
\label{sec:conclusions}
In this work we present \emph{Learn2Adapt} (\emph{L2A}), a novel rate adaptation algorithm for HAS, based on online learning. Overall, our proposed method performs well over a wide spectrum of streaming scenarios, due to its design principle; its ability to learn. It does so  without requiring any parameter tuning, modifications according to application type or statistical assumptions for the channel. The \emph{robustness} property of \emph{L2A} allows it to be classified in the small set of rate adaptation algorithms for video streaming, that mitigate the main limitation of existing mobile \ac{HAS} approaches; the dependence on statistical models for the unknowns. This is of significant relevance in the field of modern \ac{HAS}, where \ac{OTT} video service providers are continuously expanding their services to include more diverse user classes, network scenarios and streaming applications.

    \acrodef{16-QAM}{16 Quadrature Amplitude Modulation}
    \acrodef{64-QAM}{64 Quadrature Amplitude Modulation}
    \acrodef{256-QAM}{256 Quadrature Amplitude Modulation}
    \acrodef{ACD}{Asymmetric Cooperation Diversity}
    \acrodef{ACF}{Autocorrelation Function}
    \acrodef{ACK}{Acknowledgment}
    \acrodef{ARQ}{Automatic Repeat Request}
    \acrodef{AP}{Access Point}
    \acrodef{aDA}{advanced Dynamic Algorithm}
    \acrodef{ADC}{Analog-to-Digital Converter}
    \acrodef{APP}{Application layer}
    \acrodef{ASIC}{Application Specific Integrated Circuits}
    \acrodef{AWGN}{Additive White Gaussian Noise}
    \acrodef{bDA}{basic Dynamic Algorithm}
    \acrodef{BER}{Bit Error Rate}
    \acrodef{BPSK}{Binary Phase Shift Keying}
    \acrodef{BS}{Base Station}
    \acrodef{CC}{Coded Cooperation}
    \acrodef{CBR}{Constant Bit Rate}
    \acrodef{CDF}{Cumulative Distribution Function}
    \acrodef{CDMA}{Code-Division-Multiple-Access}
    \acrodef{CIC}{Cascaded Integrator-Comb}
    \acrodef{CIF}{Common Intermediate Format}
    \acrodef{CNR}{Channel Gain-to-Noise Ratio}
    \acrodef{CRC}{Cyclic Redundancy Check}
    \acrodef{CTS}{Clear-To-Send}
    \acrodef{CSI}{Channel State Information}
    \acrodef{DAC}{Digital-to-Analog Converter}
    \acrodef{DF}{Decode-and-Forward}
    \acrodef{PSD}{Power Spectral Density}
    \acrodef{DCT}{Discrete Cosine Transform}
    \acrodef{DoF}{Degree of Freedom}
    \acrodef{DIFS}{Distributed Coordination Function IFS}
    \acrodef{DIV}{Distortion In interVal}
    \acrodef{DLC}{Data Link Control layer}
    \acrodef{DSP}{Digital Signal Processor}
    \acrodef{ebDA}{extended basic Dynamic Algorithm}
    \acrodef{eCTS}{extended CTS}
    \acrodef{ePLCP}{extended physical layer convergence procedure}
    \acrodef{eRTS}{extended request-to-send}
    \acrodef{FDMA}{Frequency Division Multiple Access}
    \acrodef{FEC}{Forward Error Correction}
    \acrodef{FER}{Frame Error Rate}
    \acrodef{FIFO}{First-In-First-Out}
    \acrodef{FPGA}{Field Programmable Gate Array}
    \acrodef{FCS}{Frame Check Sequence}
    \acrodef{GMSK}{Gaussian Minimum Shift Keying}
    \acrodef{GOP}{Group Of Pictures}
    \acrodef{GSR}{GNU Software Radio}
    \acrodef{HTTP}{Hypertext Transfer Protocol}
    \acrodef{HTML}{Hypertext Mark-up Language}
    \acrodef{ICI}{Inter-carrier Interference}
    \acrodef{IEEE}{Institute of Electrical and Electronics Engineers, Inc.}
    \acrodef{IFS}{Inter-Frame Space}
    \acrodef{IP}{Internet Protocol}
    \acrodef{IPC}{Inter-Process Communication}
    \acrodef{ISI}{Inter-symbol Interference}
    \acrodef{LLC}{Logical Link Control}
    \acrodef{MAC}{Medium Access Control}
    \acrodef{MCM}{Multi Carrier Modulation}
    \acrodef{MIMO}{Multiple-Input Multiple-Output}
    \acrodef{MISO}{Multiple-Input Single-Output}
    \acrodef{MPEG}{Moving Pictures Expert Group}
    \acrodef{MAF}{Mobile\slash{}Akiyo\slash{}Football}
    \acrodef{MIV}{Mean distortion In Interval}
    \acrodef{MOS}{Mean Opinion Score}
    \acrodef{MRC}{Maximum Ratio Combining}
    \acrodef{MSB}{Most Significant Bit}
    \acrodef{MSE}{Mean Squared Error}
    \acrodef{MSC}{Message Sequence Chart}
    \acrodef{MSS}{Maximum Segment Size}
    \acrodef{MTU}{Maximum Transmission Unit}
    \acrodef{MUD}{Multiuser Diversity}
    \acrodef{NAV}{Network Allocation Vector}
    \acrodef{NTP}{Network Time Protocol}
    \acrodef{OFDM}{Orthogonal Frequency Division Multiplexing}
    \acrodef{FDD}{Frequency Division Duplexing}
    \acrodef{OFDMA}{OFDM multiple Access}
    \acrodef{OS}{Operating System}
    \acrodef{PER}{Packet Error Rate}
    \acrodef{PDF}{Probability Density Function}
    \acrodef{PDU}{Protocol Data Unit}
    \acrodef{PHY}{Physical layer}
    \acrodef{PLCP}{Physical Layer Convergence Procedure}
    \acrodef{PPDU}{Physical Protocol Data Unit}
    \acrodef{PSNR}{Peak Signal-to-Noise Ratio}
    \acrodef{QoS}{Quality of Service}
    \acrodef{QPSK}{Quadperrature Phase Shift Keying}
    \acrodef{RCPC}{Rate-Compatible Punctured Convolutional}
    \acrodef{RF}{Radio Frequency}
    \acrodef{RTT}{Round Trip Time}
    \acrodef{RTS}{Request-To-Send}
    \acrodef{RS}{Relay Station}
    \acrodef{RSSI}{Received Signal Strength Indication}
    \acrodef{SDF}{Selection Decode-and-Forward}
    \acrodef{SDR}{Software Defined Radio}
    \acrodef{SIFS}{Short Inter-Frame Space}
    \acrodef{SR}{Software Radio}
    \acrodef{SEP}{Symbol Error Probability}
    \acrodef{SCM}{Single Carrier Modulation}
    \acrodef{SNR}{Signal-to-Noise Ratio}
    \acrodef{STC}{Space-Time Coding}
    \acrodef{STBC}{Space-Time Block Coding}
    \acrodef{TACD}{Traffic-aware Asymmetric Cooperation Diversity}
    \acrodef{TCP}{Transmission Control Protocol}
    \acrodef{TDMA}{Time Division Multiple Access}
    \acrodef{UDP}{User Datagram Protocol}
    \acrodef{USRP}{Universal Software Radio Peripheral}
    \acrodef{VoIP}{Voice over IP}
    \acrodef{VQM}{Video Queue Management}
    \acrodef{WLAN}{Wireless Local Area Network}
    \acrodef{WMAN}{Wireless Metropolitan Area Network}
    \acrodef{WT}{Wireless Terminal}
    \acrodef{WWW}{World-Wide-Web}


    \acrodef{HAS}{HTTP Adaptive Streaming}
	\acrodef{QoE}{Quality of Experience}
	\acrodef{VBR}{Variable Bit Rate}
	\acrodef{VoD}{Video on Demand}
	\acrodef{3G}{$3^{rd}$ Generation}
	\acrodef{4G}{$4^{th}$ Generation}
	\acrodef{5G}{$5^{th}$ Generation}
	\acrodef{HTTP}{HyperText Transfer Protocol}
	\acrodef{TCP}{Transmission Control Protocol}
	\acrodef{MPEG}{Moving Picture Expert Group}
	\acrodef{URL}{Uniform Resource Locator}
	\acrodef{EWMA}{Exponentially Weighted Moving Average}
	\acrodef{ICN}{Information Centric Networking }
	\acrodef{TCP/IP}{Transmission Control Protocol/Internet Protocol}
	\acrodef{CDF}{Cummulative Distribution Function}
	\acrodef{2CDF}{empirical Cummulative Distribution Function}
	\acrodef{DASH}{Dynamic Adaptive Streaming over HTTP}
	\acrodef{IRT}{Inter-Request Time}
	\acrodef{IAT}{Inter-Arrival Time}
	\acrodef{ADB}{Android Debug Bridge}
	\acrodef{QUIC}{Quick UDP Internet Connections}
	\acrodef{TLS}{Transport Layer Security}
	\acrodef{TBF}{Token Bucket Filter}
	\acrodef{LTE}{Long Term Evolution}
	\acrodef{USB}{Universal Serial Bus}
	\acrodef{SINR}{Signal-to-Interference-plus-Noise Ratio}
		\acrodef{CDN}{Content Delivery Network}
		\acrodef{DNS}{Domain Name System}
		\acrodef{DPI}{Deep Packet Inspection}
		\acrodef{UI}{User Interface}
		\acrodef{RL}{Reinforcement Learning}
		\acrodef{DP}{Dynamic programming}
		\acrodef{DL}{Deep Learning}
		\acrodef{NN}{Neural Networks}
		\acrodef{MDP}{Markov Decision Process}
		\acrodef{QL}{Q-Learning}
		\acrodef{OCO}{Online Convex Optimization}
		\acrodef{OTT}{Over-The-Top}
	


%





\ifCLASSOPTIONcaptionsoff
  \newpage
\fi



	\bibliographystyle{ieeetr}
	\bibliography{My_bib} 

\begin{thebibliography}{10}

\bibitem{cisco2017}
{Cisco Visual Networking Index}, ``{Forecast and Trends, 2017--2022},'' {\em
  White Paper}, Feb. 2019.

\bibitem{Sodagar}
I.~{Sodagar}, ``{The MPEG-DASH Standard for Multimedia Streaming Over the
  Internet},'' {\em IEEE MultiMedia}, vol.~18, April 2011.

\bibitem{ciscomobile}
{Cisco Visual Networking Index}, ``{Global Mobile Data Traffic Forecast Update,
  2017--2022},'' {\em White Paper}, Feb. 2019.

\bibitem{own}
T.~Karagkioules, C.~Concolato, D.~Tsilimantos, and S.~Valentin, ``{A
  Comparative Case Study of HTTP Adaptive Streaming Algorithms in Mobile
  Networks},'' in {\em Proc. of ACM Int. Conf. on NOSSDAV}, June 2017.

\bibitem{8048013}
M.~{Gadaleta}, F.~{Chiariotti}, M.~{Rossi}, and A.~{Zanella}, ``{D-DASH: A Deep
  Q-Learning Framework for DASH Video Streaming},'' {\em IEEE Transactions on
  Cognitive Communications and Networking}, vol.~3, Dec 2017.

\bibitem{MPC}
X.~Yin, A.~Jindal, V.~Sekar, and B.~Sinopoli, ``{A Control-Theoretic Approach
  for Dynamic Adaptive Video Streaming over HTTP},'' in {\em Proc. ACM Int.
  Conf. on SIGCOMM}, Aug. 2015.

\bibitem{Zinkevich:2003:OCP:3041838.3041955}
M.~Zinkevich, ``{Online Convex Programming and Generalized Infinitesimal
  Gradient Ascent},'' in {\em Proc. of ICML}, 2003.

\bibitem{DBLP:journals/corr/abs-1804-04529}
E.~V. Belmega, P.~Mertikopoulos, R.~Negrel, and L.~Sanguinetti, ``{Online
  convex optimization and no-regret learning: Algorithms, guarantees and
  applications},'' {\em arXiv e-prints}, 2018.

\bibitem{THOR}
N.~Liakopoulos, G.~Paschos, and T.~Spyropoulos, ``No {R}egret in {C}loud
  {R}esources {R}eservation with {V}iolation {G}uarantees,'' in {\em Proc. IEEE
  INFOCOM}, May 2019.

\bibitem{8027140}
T.~{Chen}, Q.~{Ling}, and G.~B. {Giannakis}, ``{An Online Convex Optimization
  Approach to Proactive Network Resource Allocation},'' {\em IEEE Transactions
  on Signal Processing}, vol.~65, Dec 2017.

\bibitem{HAS_QoE}
M.~{Seufert}, S.~{Egger}, M.~{Slanina}, T.~{Zinner}, T.~{Hoßfeld}, and
  P.~{Tran-Gia}, ``{A Survey on Quality of Experience of HTTP Adaptive
  Streaming},'' {\em IEEE Communications Surveys Tutorials}, vol.~17,
  Firstquarter 2015.

\bibitem{Akhshabi:2013:STS:2460782.2460786}
S.~Akhshabi, L.~Anantakrishnan, C.~Dovrolis, and A.~C. Begen, ``{Server-based
  Traffic Shaping for Stabilizing Oscillating Adaptive Streaming Players},'' in
  {\em Proc. of ACM Int. Conf. on NOSSDAV}, June 2013.

\bibitem{Krishnamoorthi:2017:BPB:3083187.3083193}
V.~Krishnamoorthi, N.~Carlsson, E.~Halepovic, and E.~Petajan, ``{BUFFEST:
  Predicting Buffer Conditions and Real-time Requirements of HTTP(S) Adaptive
  Streaming Clients},'' in {\em Proc. of ACM Int. Conf. on MMSys}, June 2017.

\bibitem{6920044}
N.~{Bouten}, S.~{Latré}, J.~{Famaey}, W.~{Van Leekwijck}, and F.~{De Turck},
  ``{In-Network Quality Optimization for Adaptive Video Streaming Services},''
  {\em IEEE Transactions on Multimedia}, vol.~16, Dec 2014.

\bibitem{Panda}
Z.~{Li}, X.~{Zhu}, J.~{Gahm}, R.~{Pan}, H.~{Hu}, A.~C. {Begen}, and D.~{Oran},
  ``{Probe and Adapt: Rate Adaptation for HTTP Video Streaming At Scale},''
  {\em IEEE Journal on Selected Areas of Communication}, Apr. 2014.

\bibitem{6704839}
J.~{Jiang}, V.~{Sekar}, and H.~{Zhang}, ``{Improving Fairness, Efficiency, and
  Stability in HTTP-Based Adaptive Video Streaming With Festive},'' {\em
  IEEE/ACM Transactions on Networking}, vol.~22, Feb 2014.

\bibitem{Netflix}
T.-Y. Huang, R.~Johari, N.~McKeown, M.~Trunnell, and M.~Watson, ``{A
  Buffer-based Approach to Rate Adaptation: Evidence from a Large Video
  Streaming Service},'' in {\em Proc. of ACM Conf. on SIGCOMM}, 2014.

\bibitem{8411495}
S.~{Kim} and C.~{Kim}, ``{XMAS: An Efficient Mobile Adaptive Streaming Scheme
  Based on Traffic Shaping},'' {\em IEEE Transactions on Multimedia}, vol.~21,
  Feb 2019.

\bibitem{Xie:2015:PPL:2789168.2790118}
X.~Xie, X.~Zhang, S.~Kumar, and L.~E. Li, ``{piStream: Physical Layer Informed
  Adaptive Video Streaming over LTE},'' in {\em Proc. of ACM Int. Conf. on
  MobiCom}, 2015.

\bibitem{IOBBA}
S.~Mekki, T.~Karagkioules, and S.~Valentin, ``{HTTP Adaptive Streaming with
  Indoors-Outdoors Detection in Mobile Networks},'' in {\em Proc. of IEEE
  INFOCOM Workshops}, May 2017.

\bibitem{BOLA}
K.~Spiteri, R.~Urgaonkar, and R.~K. Sitaraman, ``{BOLA: Near-optimal bitrate
  adaptation for online videos},'' in {\em IEEE INFOCOM}, April 2016.

\bibitem{7393865}
C.~Zhou, C.~Lin, and Z.~Guo, ``{mDASH: A Markov Decision-Based Rate Adaptation
  Approach for Dynamic HTTP Streaming},'' {\em IEEE Transactions on
  Multimedia}, April 2016.

\bibitem{7305810}
A.~{Bokani}, M.~{Hassan}, S.~{Kanhere}, and X.~{Zhu}, ``{Optimizing HTTP-Based
  Adaptive Streaming in Vehicular Environment Using Markov Decision Process},''
  {\em IEEE Transactions on Multimedia}, vol.~17, Dec 2015.

\bibitem{4083008}
M.~Claeys, S.~Latré, J.~Famaey, T.~Wu, W.~Van~Leekwijck, and F.~De~Turck,
  ``{Design of a Q-learning-based client quality selection algorithm for HTTP
  adaptive video streaming},'' in {\em In Proc. of Adaptive and Learning Agents
  Workshop, part of AAMAS}, May 2013.

\bibitem{Pensieve}
H.~Mao, R.~Netravali, and M.~Alizadeh, ``Neural adaptive video streaming with
  {Pensieve},'' in {\em Proc. of ACM Int. Conf. on SIGCOM}, 2017.

\bibitem{Huang:2019:HEV:3304109.3306219}
T.-Y. Huang, C.~Ekanadham, A.~J. Berglund, and Z.~Li, ``Hindsight: Evaluate
  video bitrate adaptation at scale,'' in {\em In Proc. of ACM MMSys}, Jun
  2019.

\bibitem{8424813}
A.~{Bentaleb}, B.~{Taani}, A.~C. {Begen}, C.~{Timmerer}, and R.~{Zimmermann},
  ``{A Survey on Bitrate Adaptation Schemes for Streaming Media Over HTTP},''
  {\em IEEE Communications Surveys Tutorials}, vol.~21, Firstquarter 2019.

\bibitem{TPclassification}
D.~Tsilimantos, T.~Karagkioules, and S.~Valentin, ``Classifying flows and
  buffer state for {YouTube's} {HTTP} adaptive streaming service in mobile
  networks,'' in {\em Proc. ACM Int. Conf. on MMSys}, June 2018.

\bibitem{COLD}
N.~Liakopoulos, A.~Destounis, G.~Paschos, T.~Spyropoulos, and P.~Mertikopoulos,
  ``{C}autious {R}egret {M}inimization: {O}nline {O}ptimization with
  {L}ong-{T}erm {B}udget {C}onstraints,'' in {\em Proc. of ICML}, June 2019.

\bibitem{2017arXiv170204783N}
M.~J. {Neely} and H.~{Yu}, ``{Online Convex Optimization with Time-Varying
  Constraints},'' {\em arXiv e-prints}, Feb. 2017.

\bibitem{Shalev-Shwartz:2012:OLO:2185819.2185820}
S.~Shalev-Shwartz, ``Online learning and online convex optimization,'' {\em
  Foundations and Trends on Machine Learning}, Feb. 2012.

\bibitem{Raca:2018:BTL:3204949.3208123}
D.~Raca, J.~J. Quinlan, A.~H. Zahran, and C.~J. Sreenan, ``Beyond throughput: A
  {4G} {LTE} dataset with channel and context metrics,'' in {\em Proc. of ACM
  Int. Conf. on MMSys}, June 2018.

\bibitem{Zabrovskiy:2018:MDD:3204949.3208140}
A.~Zabrovskiy, C.~Feldmann, and C.~Timmerer, ``Multi-codec {DASH} dataset,'' in
  {\em Proc. of ACM Int. Conf. on MMSys}, June 2018.

\end{thebibliography}

\end{document}